\documentclass{iopart}
\makeatletter
\renewcommand{\tableofcontents}{\section*{\contentsname} \@starttoc{toc}}
\makeatother
\expandafter\let\csname equation*\endcsname\reax
\expandafter\let\csname endequation*\endcsname\relax

\usepackage{bbm, amsmath, amssymb, amsthm, bm, textcomp,mathrsfs}
\usepackage[utf8]{inputenc}
\usepackage[T1]{fontenc}
\usepackage[english]{babel}

\usepackage{graphicx}
\usepackage{geometry}
\usepackage{color}
\usepackage{comment}

\theoremstyle{plain}
\newtheorem{Lem}{Lemma}

\def\ket#1{\left| #1\right>}
\def\bra#1{\left< #1\right|}

\newcommand\defn[1]{\textsl{#1}}

\newcommand\ketbra[2]{|#1\rangle\langle#2|}

\newcommand\cH{{\mathcal H}}
\newcommand\cM{{\mathcal M}}
\newcommand\cN{{\mathcal N}}

\newcommand\cD{{\mathcal D}}
\newcommand\cI{{\mathcal I}}
\newcommand\cB{{\mathcal B}}
\newcommand\cE{{\mathcal E}}

\newcommand\cP{{\mathcal P}}
\newcommand{\one}{\mbox{$1 \hspace{-1.0mm}  {\bf l}$}}
\DeclareMathAlphabet{\mathpzc}{OT1}{pzc}{m}{it}

\begin{document}

\title{Optimal quantum states for frequency estimation}

\author{F Fr\"owis$^{1,2}$, M Skotiniotis$^1$, B Kraus$^1$, W D\"ur$^1$}
\address{$^1$ Institut f\"ur Theoretische Physik, Universit\"at  Innsbruck, Technikerstra\ss e 25, 6020 Innsbruck, Austria}
\address{$^2$ Group of Applied Physics, University of Geneva, CH-1211 Geneva 4, Switzerland}
\ead{florian.froewis@unige.ch}
\date{\today}

\begin{abstract}
  We investigate different quantum parameter estimation scenarios in the presence of noise, and identify optimal probe states. For frequency estimation of local Hamiltonians with dephasing noise, we determine optimal probe states for up to 70 qubits, and determine their key properties. We find that the so-called one-axis twisted spin-squeezed states are only almost optimal, and that optimal states need not to be spin-squeezed. For different kinds of noise models, we investigate whether optimal states in the noiseless case remain superior to product states also in the presence of noise. For certain spatially and temporally correlated noise, we find that product states no longer allow one to reach the standard quantum limit in precision, while certain entangled states do. Our conclusions are based on numerical evidence using efficient numerical algorithms which we developed in order to treat permutational invariant systems.
\end{abstract}

\maketitle

\section{Introduction}
\label{sec:introduction}

How can one determine a quantity or parameter with high precision? This is one of the central questions in physics throughout all fields. In addition to the numerous practical problems in this quest, one faces some fundamental limits, for example, set by measurement statistics. Metrology as the science of measurement aims at identifying these limits, and at developing optimal schemes for estimating an unknown parameter in a given experimental setup. In the seminal paper by Caves~\cite{caves_quantum-mechanical_1981} these questions were tackled within the quantum mechanical framework. It was found that quantum metrology offers a significant advantage as compared to classical strategies, where a quadratic improvement in achievable precision is obtained owing to the use of quantum entanglement. Nowadays, this insights find widespread applications in interferometry~\cite{Holland:93, Hwang:02}, atomic clocks~\cite{Valencia:04, deburgh:05}, gravitational wave detectors~\cite{McKenzie:02,
collaboration_gravitational_2011} and frequency estimation~\cite{Wineland:92, Bollinger:96}.

In this work, we consider frequency estimation where the strength, $\omega$, of a local Hamiltonian, $H$, describing the interaction of $N$ qubits with an external field in $z$ direction should be estimated. In the noiseless case one finds that for classical strategies the achievable precision, $\delta \omega$, that measures the statistical deviation of the estimated parameter from the actual value, is given by the standard quantum limit (SQL), $\delta \omega \geq O(1/\sqrt{N})$ where $N$ is the number of probe systems. In turn, entangled states allow one to achieve Heisenberg scaling where $\delta \omega = O(1/N)$. Optimal states are readily identified as a coherent superposition of eigenstates of $H$ corresponding to the minimal and maximal eigenvalue. The optimal state corresponds to the GHZ state (after Greenberger, Horne and Zeilinger~\cite{GHZ}).

In the presence of noise and imperfections the situation changes drastically. For dephasing noise (what we will refer to as ``standard scenario'' in the following), one finds that GHZ states do not offer any advantage as compared to classical strategies with noise~\cite{huelga_improvement_1997}. Moreover, no state can reach the quadratic gain in precision and at most a constant gain factor can be achieved. This is expressed in terms of general bounds~\cite{escher_general_2011,demkowicz-dobrzanski_elusive_2012}, that show for generic kinds of noise the impossibility to reach Heisenberg scaling~\footnote{Note that, for purely transversal noise and in the case of certain correlated noise, a scaling better than the SQL can be achieved~ \cite{chaves_noisy_2013,Chin:12,matsuzaki_magnetic_2011,jeske:2013}. In some cases, even Heisenberg scaling can be achieved using quantum error correction~\cite{dur_improved_2013,kessler_quantum_2014,arrad_increasing_2014}, or by other means~\cite{jeske:2013}.}. In the 
standard scenario so-called one-axis twisted spin-squeezed states (SSS)~\cite{kitagawa_squeezed_1993} were identified to reach the optimal constant gain factor~\cite{ulam-orgikh_spin_2001} in the limit of large $N$, and it is widely believed that such states are optimal even for finite $N$.

However, the situation is far from being fully understood. On the one hand, optimal states for finite $N$ in the standard scenario are unknown, and it is not known whether states other than the SSS are optimal in the asymptotic case. On the other hand, beyond the standard scenario where one considers different kinds of noise models (e.g., depolarising rather than dephasing noise), or different Hamiltonians~\cite{RJ09}, basically nothing is known about optimal states be it for finite $N$ or in the asymptotic limit (an exception is transversal noise considered in~\cite{chaves_noisy_2013}).

In this paper, we are concerned with the question of identifying optimal states for frequency estimation in different noisy metrology scenarios. Our central results for the standard scenario are as follows:
\begin{itemize}
\item We determine optimal states for finite $N$. To this aim we develop numerical tools to treat permutational invariant states that allow us to investigate systems of up to 70 qubits.
\item We find that spin-squeezing is not necessary for optimality. For finite $N$ one-axis twisted SSS are in fact only almost optimal. 
\item We identify a key feature for optimality which is a specific global distribution of coefficients of the state in the eigenbasis of $H$. This also implies a certain value of the variance of $H$.
\end{itemize}
Beyond the standard scenario, where we consider local depolarising noise or spatially and temporally correlated noise, we find:
\begin{itemize}
\item Contrary to the standard scenario GHZ states may remain superior to product states also in the noisy case.
\item For spatially and temporally correlated noise, $N$ qubit product states do not even reach the SQL, whereas GHZ states do. This opens again a gap in the scaling. In addition, the equivalence between parallel and sequential strategies~\cite{maccone_intuitive_2013} (i.e., to either use one particle $N$ times, or $N$ particles in parallel) does not hold in this case.
\end{itemize}

The paper is organised as follows. In section~\ref{sec:physical-scenarios}, we provide relevant background information, and describe the different scenarios we consider. There, we also discuss the ansatz space of qubit states we consider and outline the numerical method we develop. In section~\ref{sec:standard-scenario}, we consider the standard scenario with local Hamiltonian and dephasing noise, and determine optimal states and their features. In section~\ref{sec:beyond-stand-scen}, we present results for frequency estimation scenarios with different kinds of noise models. We summarise and conclude in section~\ref{sec:summary-conclusions}. Some technical details, in particular regarding the numerical algorithm we develop here, as well as additional analytical results are presented in the appendices.

\section{Preliminaries \label{sec:physical-scenarios}}

In this section we briefly review the quantum metrology scenarios considered in this work, as well as the main theoretical tools behind our numerical routines.  Specifically, in section~\ref{sec:resource-counting} we review the main results of classical and quantum metrology and provide the mathematical descriptions for all the models we investigate.  Section~\ref{sec:ansatz-space} introduces the ansatz state space used in our numerical optimisation routines, and section~\ref{sec:numerical-methods-1} provides a brief overview of the mathematical tools on which our numerical optimisation is based on.

\subsection{Dynamical evolution models for quantum frequency estimation}
\label{sec:resource-counting}

In a metrological scenario the goal is to estimate a parameter, $\omega\in\mathbb{R}$, from a probability distribution, $p(x|\omega)$, where $x\in\mathbb{R}$ denotes the outcomes of a suitable measurement. An \defn{unbiased} estimate, $\hat{\omega}$, of $\omega$ is obtained by suitable post-processing of the measurement outcomes~\footnote{An estimate is unbiased if its expected value, $\langle\hat{\omega}\rangle$, with respect to the probability distribution $p(x|\omega)$ is equal to $\omega$.}. The variance, $\delta\omega^2=\langle  (\omega-\frac{\hat{\omega}}{\left|\mathrm{d}\hat{\omega}/\mathrm{d}\omega\right|})^2\rangle$, of an unbiased estimator is lower-bounded via the well-known Cram\'{e}r-Rao inequality~\cite{Cramer:61}
\begin{equation}
 \delta\omega^2\geq\frac{1}{n F},
 \label{CCR}
\end{equation}
where $F$ is the \defn{Fisher information} given by~\cite{Fisher:22}
\begin{equation}\label{eq:1}
 F=\int\left(\frac{\partial \ln p(x|\omega)}{\partial\omega}\right)^2\mathrm{d}x,
\end{equation}
and $n$ is the number of measurements. It is known that the lower bound in equation \eqref{CCR} can be achieved in the limit $n\to\infty$ by the maximum likelihood estimator~\cite{Fisher:25}.

Assume that the parameter of interest is a quantity governing the evolution of a physical system, such as the frequency of rotation of a spin about a given axis. Then, the scenario described above is implemented by preparing the system in some initial known state, $\rho$, and allowing it to evolve under the requisite dynamics for some time before measuring its final state, $\rho(\omega, t)$. This process is repeated $n$ times to obtain the measurement statistics, $p(x|\omega)$, from which an estimate $\hat{\omega}$ can be extracted. If the physical system is quantum mechanical in nature, i.e.~$\rho(\omega, t)\in\cB(\cH)$, then the probability distribution, $p(x|\omega)$, is given by $p(x|\omega)=\mathrm{Tr}(M_x\rho(\omega, t))$, where the set of measurement operators $\{M_x:\cB(\cH)\to\cB(\cH)\}$ satisfies $\sum_x M_x^\dagger M_x=I$.

As any set of measurement operators, $\{M_x\}$, constitute an admissible measurement it is natural to ask which measurement minimises equation \eqref{CCR} or, equivalently, maximises the Fisher information. The Quantum Fisher Information (QFI), $\mathcal{F}$, is defined as the maximal Fisher information over all allowable measurements, and is given by~\cite{H76, H80, BC94}
\begin{equation}
\mathcal{F}=\mathrm{Tr}\left[\dot{\rho}L_{\omega}\right],
\label{qfi}
\end{equation}
where
\begin{equation}
L_{\omega}=2\sum_{i,j}\frac{\bra{\psi_j}\dot{\rho}(\omega, t)\ket{\psi_i}}{\lambda_i+\lambda_j}\ket{\psi_j}\bra{\psi_i}
\label{Eq:SLD}
\end{equation}
is known as the \defn{symmetric logarithmic derivative}.  Here, $\dot{\rho}(\omega, t)=\partial\rho(\omega, t)/\partial\omega$, $\lambda_j$ are the eigenvalues of $\rho(\omega, t)$, $\ket{\psi_j}$ the corresponding eigenvectors, and the sum in equation \eqref{Eq:SLD} is over all $i,j$ satisfying $\lambda_i+\lambda_j\neq0$. The most informative measurement is the one whose elements are the projectors on the eigenspaces of the symmetric logarithmic derivative. Substituting equation \eqref{qfi} in place of the Fisher information in equation \eqref{CCR} yields the \defn{quantum Cram\'{e}r-Rao inequality} which provides the ultimate lower bound on precision achievable by a quantum mechanical strategy.

As the QFI already incorporates the optimisation over all measurements, it remains to minimise the uncertainty in $\omega$ with respect to all other available resources.  In this work we focus solely on estimating the frequency of rotation of a spin around a given axis.  For this task the two relevant resources are the number of probe systems, $N$, in a given run of the experiment and the total running time, $T=n t$, of the experiment. Here, $t$ denotes the \defn{interrogation time}, the time interval the $N$ probes are subjected to the $\omega$-dependent dynamical evolution before they are measured.  Note that in order to ensure that  the Cram\'er-Rao bound is achievable, we require $n$ in equation (\ref{CCR}) to be very large which implies that $T$ needs to be much larger than any other time scale. With respect to these resources the quantum Cram\'{e}r-Rao bound for frequency estimation reads~\cite{huelga_improvement_1997}
\begin{equation}
\delta\omega^2\geq\frac{1}{T\frac{\mathcal{F}(\rho(\omega, t))}{t}},
\end{equation}
where $\mathcal{F}(\rho(\omega, t))$ denotes the QFI of the final state of the $N$ probes. The interrogation time, $t$, is a controllable parameter which needs to be optimised in order to maximise the precision in frequency estimation.  Henceforth, for a given initial state, $| \psi \rangle $, the maximal QFI is defined as $\mathfrak{F}_{\psi} \mathrel{\mathop:}= \max_t T\mathcal{F}(\rho(\omega, t))/t$.

The dynamical evolution of the $N$ probe systems is described by a master equation of the Lindblad form
\begin{equation}
\label{eq:2}
\frac{d \rho}{dt} = -i \omega [h,\rho] + \mathcal{L}[\rho],
\end{equation}
where $\omega$ is the frequency we are interested in estimating, $\rho$ is the initial state of the $N$ probes, and $H = \omega h$ and $\mathcal{L}$ are the Hamiltonian and Lindblad operators generating the unitary and non-unitary (noisy) part of the evolution respectively.  In this work we consider the \defn{local} generator
\begin{equation}
\label{eq:49}
h\equiv S_z \mathrel{\mathop:}= \frac{1}{2}  \sum_{i=1}^N \sigma_z^{(i)}.
\end{equation}
We note that for local Hamiltonians, and in the absence of noise, the optimal precision in estimating frequency using an initially pure product state scales at the SQL, $\delta\omega^2= O(N^{-1})$. Contradistinctively, if the $N$ probes are initially prepared in a GHZ state, i.e.~an equal superposition of the eigenstates corresponding to the maximum and minimum eigenvalue of $h$, then the precision in estimation scales at the Heisenberg limit, $\delta\omega^2=O(N^{-2})$~\cite{GLM04}.

The noise models we consider here are of two main types. The first type of noise we consider is local, uncorrelated noise described by the Lindblad operator $\mathcal{L}[\rho]=\sum_{j=1}^N\mathcal{L}_j[\rho]$ with
\begin{equation}
{\cal L}_{j}[\rho]=\frac{\gamma}{2}(-\rho+\mu_x\sigma_x^{(j)}\rho\sigma_x^{(j)}+\mu_y\sigma_y^{(j)}\rho\sigma_y^{(j)}+\mu_z\sigma_z^{(j)}\rho\sigma_z^{(j)}),
\label{localnoise}
\end{equation}
where $\gamma$ denotes the strength of the noise and $\sum_i \mu_i = 1$. The case $\mu_x=\mu_y=0, \, \mu_z=1$ corresponds to the case of local, uncorrelated dephasing noise and is the main type of noise considered in this work (section~\ref{sec:standard-scenario}), whereas the case $\mu_x=\mu_y=\mu_z=1/3$ corresponds to local uncorrelated, depolarising noise (section~\ref{sec:local-depolarization}).

The second type of noise we consider is correlated dephasing noise, where correlations are both in space and time (section~\ref{sec:correlated-noise}), with Lindblad operator given by
\begin{equation}
\label{eq:53}
\mathcal{L}[\rho]\equiv -\gamma f(t) \left[ S_z, \left[ S_z,\rho \right] \right].
\end{equation}
Here, $f(t)=1 -\exp(-\gamma t)$ denotes the temporal profile of the noise, whereas the double-commutator in equation \eqref{eq:53} governs the spatial correlations.  Such type of noise is physically relevant in ion trap setups due to fluctuations in the phase reference beams addressing all $N$ ions collectively, i.e.~with infinite spatial correlation length, but finite memory~\cite{kielpinski_decoherence-free_2001,langer_long-lived_2005,roos_designer_2006,monz_14-qubit_2011}.  
It has been demonstrated that for local Hamiltonians temporal correlations alone allow for sub-SQL precision in optical interferometry~\cite{Szankowski:12}, whereas spatially correlated noise alone, with finite or infinite correlation length, can even allow for Heisenberg scaling in precision using a suitably chosen higher dimensional entangled state~\cite{jeske:2013} or considering a different Hamiltonian~\cite{dorner_quantum_2012}.

A key feature for all scenarios considered in this paper is that the unitary and noisy part in equation (\ref{eq:2}) commute. This means that the solution for the time-dependent density operator reads
\begin{equation}
\rho(\omega,t) = e^{-i \omega h t} \mathcal{E}[\left| \psi \right\rangle\!\left\langle \psi\right| ] e^{i \omega h t},
\label{eq:3}
\end{equation}
where $\mathcal{E}[\rho]  = \exp( \mathcal{L}t)[\rho]$ denotes the noise channel and $| \psi \rangle $ the initial state. Writing $\mathcal{E}[\left| \psi \right\rangle\!\left\langle \psi\right| ]$ in its spectral decomposition, $\mathcal{E}[\left| \psi \right\rangle\!\left\langle \psi\right| ] = \sum_i \lambda_i \left| \psi_i \right\rangle\!\left\langle \psi_i\right| $, the QFI can be expressed as~\cite{H76, BC94}
\begin{equation}
\label{eq:4}
\mathcal{F} = 4 t^2\sum_{i<j} \frac{(\lambda_i-\lambda_j)^2}{\lambda_i+\lambda_j} \left| \left\langle \psi_i \right|  h \left| \psi_j \right\rangle  \right|^2.
\end{equation}
Note that $\mathcal{F}$ is here independent of $\omega$. For the case of pure states and in the absence of noise it can be easily verified that $\mathcal{F} = 4 t^2 \mathcal{V}(h) \mathrel{\mathop:}= 4t^2 (\langle h^2 \rangle - \langle h \rangle^2)$.
However, even if we start with an initially pure state, the evolution of such a state under the full dynamics of equation \eqref{eq:2} will in general yield a mixed state. We emphasise that maximising $\mathcal{F}/t$ in time reduces the role of $\gamma$ to a proportionality constant, that is, $\mathfrak{F}_{\psi} = \mathfrak{F}_{\psi}|_{\gamma = 1}/\gamma$. 
Therefore, the findings presented in this paper are independent of a specific choice of the value of $\gamma$, which eases also the numerical effort. This is in contrast to the so-called phase estimation scenario, where the interrogation time is fixed. Then, there are typically two regimes, $N \ll \gamma$ and $N \gg \gamma$, where one can observe qualitatively different behavior for the optimal initial states \cite{escher_general_2011,demkowicz-dobrzanski_elusive_2012}.

\subsection{Ansatz space}
\label{sec:ansatz-space}

Our goal is to numerically determine the initial states leading to maximal sensitivity for the case of noisy frequency estimation and to identify the key properties of such states.  Specifically, we are interested in determining the states that maximise the QFI.  As the latter is convex~\cite{BCM96} it suffices to consider only pure states of $N$ qubits.  However, numerically searching for arbitrary pure states that maximise the QFI seems unfeasible for a large number of qubits. Thus, we restrict ourselves to a subspace of the total Hilbert space in order to reduce computational effort.

In order to pick the most suitable subspace we note that the dynamical evolutions we consider are symmetric under particle exchange. This observation has motivated several authors to consider pure initial states that are symmetric under particle permutations~\cite{huelga_improvement_1997, dorner_quantum_2012, Dobrzanski:09}, i.e.~states of the form
\begin{equation}
\label{eq:5}
\left| \psi \right\rangle = \sum_{m = -j_{\max}}^{j_{\max}} c_m \left| j_{\max}, m \right\rangle.
\end{equation}
The states $| j_{\max},m \rangle $ are simultaneous eigenstates of the total angular momentum operator, $S^2$, and its projection onto the $z$-axis, $S_z$, with corresponding eigenvalues $j_{\max}(j_{\max} +1) = N/2(N/2+1)$ and $m$ respectively, and are the so-called Dicke states with $N/2+m$ excitations~\cite{dicke_coherence_1954}. Note that we define $\ket{0}$ as the excited state, so that for instance the Dicke state of 3 qubits, one of which is in the $|0 \rangle $ state, is $\ket{3/2,-1/2}=1/\sqrt{3}(\ket{011}+\ket{101}+\ket{110})$.

We can reduce the number of parameters in equation \eqref{eq:5} further by requiring that the coefficients $c_m$ are real and positive, as any complex phase can be taken into account by applying a phase gate that commutes with the dynamics (equation (\ref{eq:2})). In addition, the dynamics are invariant under collective spin flips $\sigma_x^{\otimes N}$. Requiring the same symmetry for the states leads to $c_m = c_{-m}$.

With the exception of correlated noise (see equation \eqref{eq:53})~\cite{dorner_quantum_2012}, we are not aware of any proof that the state that maximises the QFI in the presence of local uncorrelated noise  belongs to our ansatz space.  However, by extensive numerical studies for $N = 2,3$, and by comparing specific examples of asymmetric states and their symmetrised counterparts, it seems that the optimal state must exhibit the same symmetry as the dynamics.  With all these considerations taken into account we choose our ansatz space as the space spanned by the states of equation \eqref{eq:5} with $c_m\in\mathbb{R}$, $c_m>0$, and $c_m=c_{-m}$. This subspace is denoted by $\mathcal{S}_N$ in the following. Note that with these restrictions, states in equation \eqref{eq:5} point in $x$ direction, i.e., with $S_x = 1/2 \sum_{i = 1}^N \sigma_x^{(i)},\,S_y = 1/2 \sum_{i = 1}^N \sigma_y^{(i)}$, one has $\langle S_x \rangle\geq 0,\,\langle S_y \rangle=\langle S_z\rangle=0$.

Within this ansatz space, there are several state families that are of particular interest. We now define four such state families for comparison to the optimal states for noisy frequency estimation.
The first family of states we examine are the product states.  The optimal product state for all considered scenarios is $\left| \mathrm{PS} \right\rangle =| + \rangle ^{\otimes N}$ which, when expressed in the $\{\ket{ j_{\max}, m }\}$ basis, reads
\begin{equation}
\ket{\mathrm{PS}}=\sum_{m = -j_{\max}}^{j_{\max}}\sqrt{\binom{N}{N/2+m}/2^N}\ket{j_{\max}, m}.
\label{productstate}
\end{equation}
The second family of states we consider is the GHZ state
\begin{equation}
\ket{\mathrm{GHZ}}=\frac{1}{\sqrt{2}}\left(\ket{N/2,-N/2}+\ket{N/2,N/2}\right).
\label{GHZstate}
\end{equation}
As mentioned in section~\ref{sec:resource-counting} these states achieve Heisenberg scaling in precision for noiseless frequency estimation using the  Hamiltonian in equation \eqref{eq:49}.

The third family of states we consider are the one-axis twisted SSS~\cite{kitagawa_squeezed_1993}
\begin{equation}
\label{eq:14}
\left| \mathrm{SSS}(\mu) \right\rangle = e^{-i \nu S_x} e^{-i \mu S_z^2} \left| \mathrm{PS}\right\rangle.
\end{equation}
This state family is defined by two parameters:  the squeezing parameter, $\mu$, and a local rotation parameter $\nu$ which serves to re-orient the squeezing axis so that the benefit for the specific  Hamiltonian is optimised. As the value of $\nu$ depends only on $\mu$, the one-axis twisted SSS are essentially a one-parameter family. It was shown that for local Hamiltonians and for local uncorrelated dephasing noise (see section~\ref{sec:standard-scenario}) one axis-twisted SSS are asymptotically optimal~\cite{ulam-orgikh_spin_2001}. 

The fourth family of states we consider are the Dicke states in the $x$-basis,
\begin{equation}
\left| D_k \right\rangle  = \mathit{Had}^{\otimes N} \left| N/2, N/2-k\right\rangle,
\label{Dickestates}
\end{equation}
where $\mathit{Had}$ denotes the Hadamard operation
\begin{equation}
\mathit{Had}=\frac{1}{\sqrt{2}}\begin{pmatrix} 1& 1\\1&-1\end{pmatrix}.
\end{equation}
Note that $| \mathrm{SSS}(0) \rangle  = \left| D_0 \right\rangle  = \left| \mathrm{PS} \right\rangle$.

In addition to the four families of states above, we conduct a numerical search for the optimal state in the entire ansatz space $\mathcal{S}_N$. For this task, we choose the Nelder-Mead simplex algorithm~\cite{nelder_simplex_1965} which is known to be successful for low-dimensional, unconditioned search problems~\cite{lagarias_convergence_1998}. This algorithm takes as input a state and time and optimises the coefficients $c_m$ of equation \eqref{eq:5} by maximising $\mathcal{F}/t$. The input state is either a specific state or a randomly chosen one. The coefficients, $c_m$, of a random state are stochastic variables which follow a normal distribution. After all coefficients are chosen the random state is normalised. Such a random state is denoted by $\left| \psi_{\mathrm{rand}} \right\rangle $.

In the next subsection we describe how the QFI can be efficiently calculated for states in our ansatz space.

\subsection{Numerical methods}
\label{sec:numerical-methods-1}

Evaluating the QFI (see equation \eqref{eq:4}) is computationally hard as it requires full diagonalisation of the density matrix $\mathcal{E}[\left| \psi \right\rangle\!\left\langle \psi\right|]$, where $\cE$ denotes the noisy channel (see equation \eqref{eq:3}). However,  if the eigenvalues of $\mathcal{E}[\left| \psi \right\rangle\!\left\langle \psi\right| ]$ are highly degenerate, then the computational effort for calculating the QFI can be significantly reduced.  A state that is symmetric under particle permutations exhibits such large degeneracies. In particular, for $h = S_z$ and $\mathcal{L}$ as in equations (\ref{localnoise},~\ref{eq:53}), we are able to use a specific representation of $\rho(t)$ and $h$ that allows for an efficient calculation of $\mathcal{F}$. Here, we briefly summarise this representation for the local uncorrelated dephasing scenario, equation \eqref{localnoise}, with $\mu_x=\mu_y=0, \, \mu_z=1$, and refer the reader to \ref{sec:numerical-methods} for more details.

Under the effect of local dephasing, any state in the ansatz space remains permutationally invariant. Hartmann showed in~\cite{hartmann_generalized_2012} that any permutationally invariant state can be represented by a weighted sum of $O(N^3)$ specific operators. Thus, even though $\mathcal{E}[\left| \psi \right\rangle\!\left\langle \psi\right|]$ may be of full rank, where $\ket{\psi}$ is permutationally symmetric, its representation is always efficient. Moreover, the operators introduced in~\cite{hartmann_generalized_2012} are such that the action of local dephasing is easy to express. As shown in \ref{sec:numerical-methods} the spectral decomposition of $\mathcal{E}[\left| \psi \right\rangle\!\left\langle \psi\right|]$ can then be efficiently computed.
For an efficient computation of $\mathcal{F}$, it is necessary that $h$ obeys the same symmetries as the state $\mathcal{E}[\left| \psi \right\rangle\!\left\langle \psi\right|]$, as $\mathcal{F}\left(U\mathcal{E}[\left| \psi \right\rangle\!\left\langle \psi\right|]U^\dagger\right)=\mathcal{F}\left(\mathcal{E}[\left| \psi \right\rangle\!\left\langle \psi\right|]\right)$.

The general procedure for computing the QFI is hence the following. Express the initial state, equation \eqref{eq:5}, in terms of the ``Hartmann operators'' introduced in~\cite{hartmann_generalized_2012}. Next, calculate the action of local dephasing on the basis operators. Then, express the operators in terms of joint eigenvectors of $S^2$ and $S_z$, and numerically determine the spectral decomposition of $\mathcal{E}[\left| \psi \right\rangle\!\left\langle \psi\right|]$ as  a function of the coefficients, $c_m$. One can then efficiently evaluate equation \eqref{eq:4}.  We note that the same method can be applied for the case of depolarising noise, even though the effect of this noise model when computed in terms of the ``Hartmann operators'' is more cumbersome. Spatially correlated noise, on the other hand, is easier to treat since the states of equation (\ref{eq:5}) stay within the subspace spanned by $\left\{ \left| j,m \right\rangle  \right\}_m$ under the action of the noise (see \ref{sec:numerical-methods} 
for more details).

\section{Local Hamiltonian with dephasing noise}
\label{sec:standard-scenario}

In this section, we consider the standard scenario where $h$ is given by equation \eqref{eq:49} and $\mathcal{L}$ is given by equation \eqref{localnoise} with $\mu_x=\mu_y=0, \,\mu_z=1$, respectively. Hence, the unitary and noisy evolutions in the standard scenario act parallel to each other if visualised on the Bloch sphere.

Let us first recall the known results pertaining to the standard scenario. It is easy to see that the product state, $\ket{+}^{\otimes N}$, leads to $\mathfrak{F}_{\mathrm{PS}} = N T/(2 \gamma e)$, which has the same scaling as in the absence of noise. The GHZ state, however, yields
$\mathfrak{F}_{\mathrm{GHZ}} = \mathfrak{F}_{\mathrm{PS}}$, which is qualitatively different from the quadratic scaling in $N$ for $\gamma = 0$~\footnote{There is an apparent contradiction in these statements: if $\gamma \rightarrow 0$, then $\mathfrak{F}$ of any state is infinite, while for $\gamma = 0$, one often discusses the finite value (e.g., $\mathcal{F} = t^2 N^2$ for the GHZ state). However, the optimal measurement time for any state is of the order of $1/ \gamma$. That is, the measurement time goes to infinity for vanishing $\gamma$. For infinite measurement times and $\gamma = 0$, the 
QFI of any nontrivial state is infinite.}. Hence, the advantage of the GHZ state compared to the product state is lost as soon as $\gamma \neq 0$.
As already mentioned in section~\ref{sec:resource-counting}, it is known that, asymptotically, the maximal QFI,$\mathfrak{F}$, for any state scales at most linearly with the system size, and the maximal improvement, for any $| \psi \rangle $ and $N$, is bounded by $\mathfrak{F}_{\psi} \leq NT/(2\gamma) = e \mathfrak{F}_{\mathrm{PS}}$, where the bound can only be achieved for $N\rightarrow \infty$~\cite{escher_general_2011}. Thus, the minimum error obtainable is given by $0.61 (\delta \omega)_{\mathrm{PS}}$. It is also known that one-axis twisted SSS with a specific squeezing parameter $\mu$ (which depends on $N$ only) achieve this bound asymptotically~\cite{ulam-orgikh_spin_2001}.

Let us now consider the finite $N$ case. Our aim here is to numerically determine the optimal states and identify their properties. First, in section~\ref{sec:optim-init-stat}, we present numerical evidence that we are able to find the global optimal initial state, $| \psi_{\mathrm{opt}} \rangle $, with very high accuracy. Then, we reveal that there exists a $\mu = \mu_{\mathrm{opt}}$ such that $| \mathrm{SSS}(\mu_{\mathrm{opt}}) \rangle $ is already almost optimal for finite system sizes. However, this implies that there exist states that perform slightly better than $| \mathrm{SSS}(\mu_{\mathrm{opt}}) \rangle $.
Next, in section~\ref{sec:spin-squeezing-not}, we investigate the squeezing strength and the variance $\mathcal{V}(h)$ of the optimal states. Surprisingly, our results indicate that squeezing, often considered as the key-property for high sensitivity experiments, is not necessary. In contrast, we find that for all optimal states it holds that the variance of $h$ scales as $\mathcal{V}(h) \propto N^{1.674}$.
Finally, in section \ref{sec:typic-ansatz-stat}, we show that a randomly chosen state $\left| \psi_{\mathrm{rand}}  \right\rangle\in \mathcal{S}_N$ leads to a $\mathfrak{F}$ that is typically larger than $\mathfrak{F}_{\mathrm{PS}}$. 

Before we go to the details, we point to figure \ref{fig:VarVsQFImax_N10_locphase} (a), which intends to give the reader an overview of the relation between the maximal QFI, $\mathfrak{F}$, and the variance, $\mathcal{V}(h)$, for $N = 10$. One observes that the performances $| \mathrm{SSS}(\mu_{\mathrm{opt}}) \rangle $ and $| \psi_{\mathrm{opt}} \rangle $ are comparable, while Dicke states are sub-optimal. In addition, the state density of a random sampling indicates the distribution of $\mathcal{S}_N$.
Similar observations can be made for all $N$ considered here. 

\begin{figure}[htbp]
  \begin{picture}(420,155) \put(0,0){\includegraphics[width=.47\columnwidth]{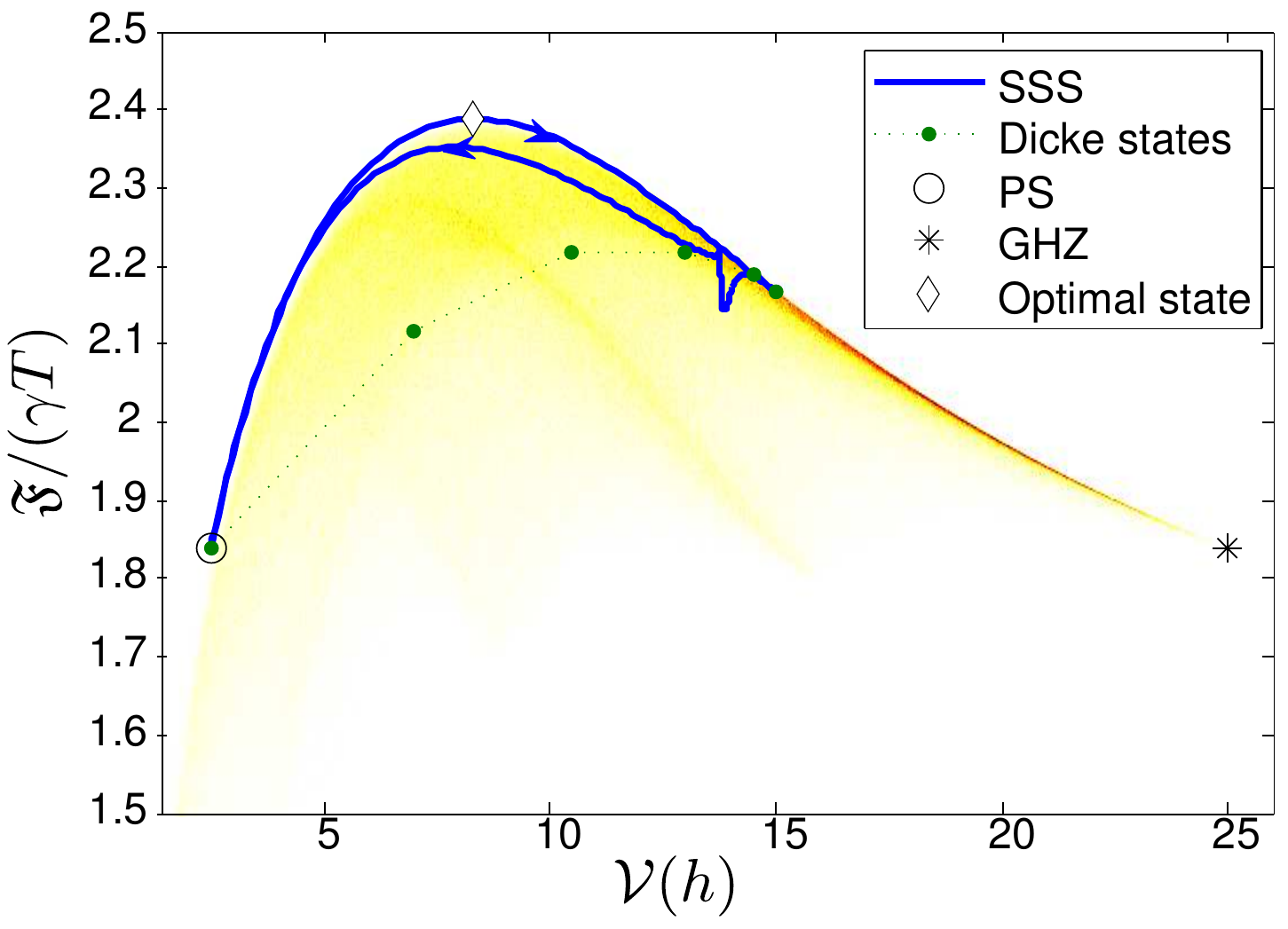}} \put(220,1){\includegraphics[width=.51\columnwidth]{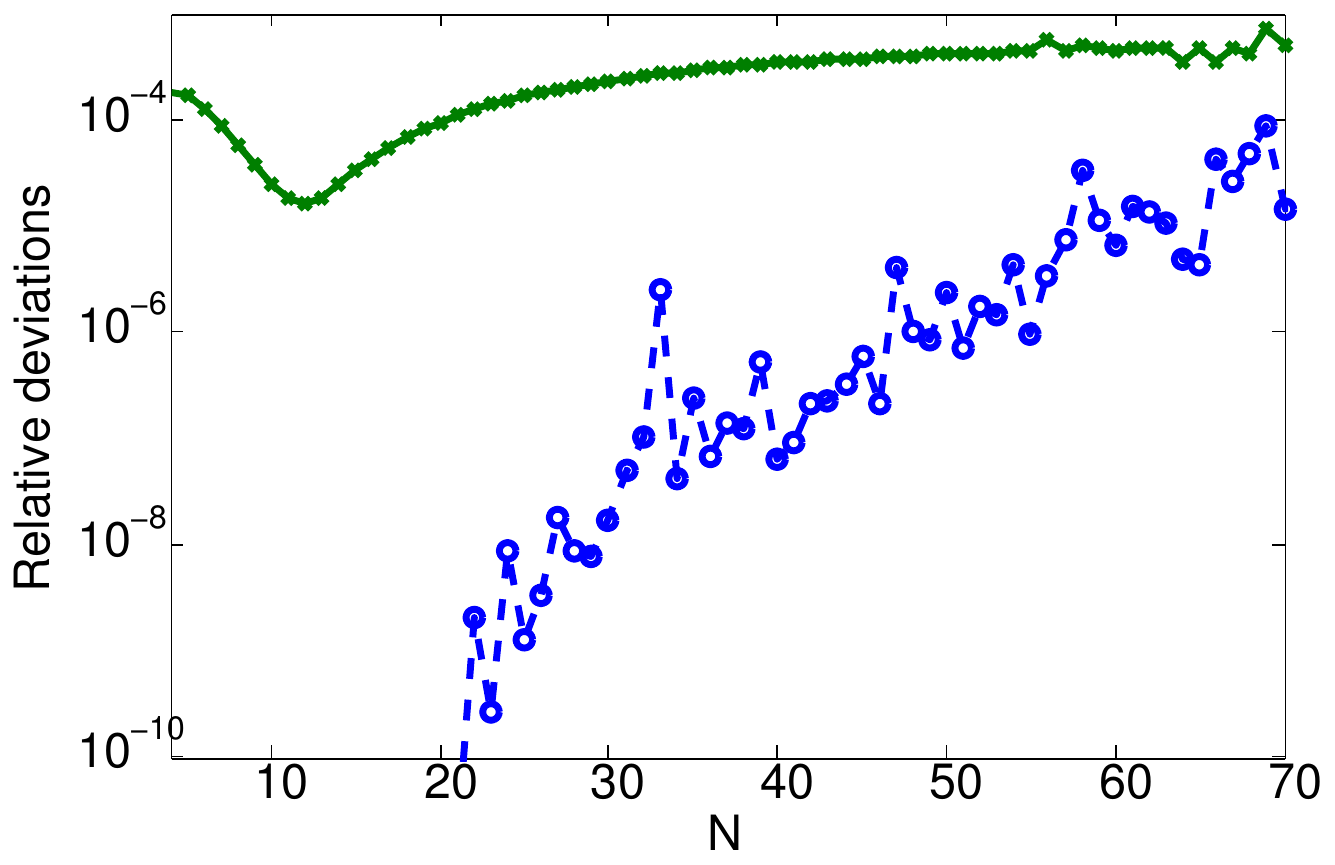}} 
    \put(0,150){(a)}
\put(220,150){(b)}
\end{picture}
\caption[]{\label{fig:VarVsQFImax_N10_locphase}(a) Variance of $h$ compared to $\mathfrak{F}/(\gamma T)$ for $N = 10$. The product state (circle) and GHZ state (star) yield the same maximal QFI, $\mathfrak{F}$, with different variances given by $N/4$ and $N^2/4$ respectively. The blue, continuous curve represents $\mathfrak{F}_{\mathrm{SSS}(\mu)}$ parametrised by $\mu \in \left[ 0,\pi \right]$, which passes through $\max_{\psi \in \mathcal{S}_N} \mathfrak{F}_{\psi}$ (diamond). Contradistinctively, the Dicke states $| D_k \rangle $ with $k \in \left\{ 0,...,\lfloor N/2 \rfloor \right\}$ follow a ``curve'' much below the maximal $\mathfrak{F}$. The results of the random sampling of 10$^6$ states from  $\mathcal{S}_N$ are indicated as a density plot in yellowish to reddish colours. Most of the sampled states are found to be in the region to the right of the global maximum. \label{fig:Fluctuations} (b) Blue, dashed line: Relative fluctuations of $\mathfrak{F}$ found by the maximisation algorithm~\cite{nelder_simplex_1965} within the eleven runs for each $N$. For $N < 22$, the fluctuations are below $10^{-9}$.  Green, solid line: Relative improvement of the numerically found optimal states with respect to $| \mathrm{SSS}(\mu_{\mathrm{opt}}) \rangle $. }
\end{figure}

\subsection{Optimal states}
\label{sec:optim-init-stat}

In this section, we present evidence that we are able to numerically determine the optimal initial states in the standard scenario.

The search for the optimal one-axis twisted SSS for a given $N$ is numerically relatively simple as one only needs to optimise over two parameters, the interrogation time, $t$, and the squeezing parameter, $\mu$. On the contrary, optimisation within the ansatz class $\mathcal{S}_N$ is much more demanding as one has to optimise over $\lfloor N/2 \rfloor +1$ parameters including time (see equation (\ref{eq:5})). 
For each value of $N$, ten randomly chosen states $| \psi_{\mathrm{rand}} \rangle \in \mathcal{S}_N$ are chosen as input states and the optimisation algorithm tries to come as close as possible to $\mathfrak{F}_{\mathrm{opt}} \mathrel{\mathop:}= \max_{\psi\in \mathcal{S}_N} \mathfrak{F}_{\psi}$. Similarly, an eleventh run is performed with $| \mathrm{SSS}(\mu_{\mathrm{opt}}) \rangle $ as the input state for the optimisation algorithm.

We first analyse the fluctuations of the results within these eleven runs for each $N$ (see figure \ref{fig:Fluctuations} (b), blue dashed line). One only finds small relative deviations from the mean result. We are thus confident that we found the global maximum within $\mathcal{S}_N$ for $N\leq 70$ as, if the algorithm would end up in local minima, one would expect to obtain significantly different values for $\mathfrak{F}$ from run to run. In the following, we refer to \textit{optimal states}, $| \psi_{\mathrm{opt}} \rangle $, as those states that are numerically found by this algorithm.

Next, we compare $\mathfrak{F}_{\mathrm{opt}}$ to $\mathfrak{F}_{ \mathrm{SSS}(\mu_{\mathrm{opt}})} $.  For all $N\leq 70$, we found a relative improvement of approximately $10^{-5}$ to $10^{-3}$, which is significantly larger than the fluctuations of the optimisation algorithm (see figure \ref{fig:Fluctuations} (b), green solid line). We therefore conclude that whereas one-axis twisted SSS are very close to the optimal states, there exist states which perform better.

\subsection{Properties of optimal states}
\label{sec:spin-squeezing-not}\label{sec:vari-as-import}

With the confidence of discussing the actual global optimal states, we can analyse their properties. Here, this is done for the squeezing strength, $\xi^2$, and for the variance $\mathcal{V}(h)$. 

The squeezing strength \cite{wang_spin_2003,sorensen_many-particle_2001} measures how spin-squeezed an $N$-qubit state is. For states in the ansatz class (equation \eqref{eq:5}), it is defined as
\begin{equation}
\label{eq:6}
\xi^2 \mathrel{\mathop:}= \min_{\varphi}\frac{\mathcal{V}[\cos( \varphi) S_z + \sin( \varphi) S_y]}{N/4}.
\end{equation}
A state is called spin-squeezed if $\xi < 1$. Let us remark that the optimal one-axis twisted SSS, $| \mathrm{SSS}(\mu_{\mathrm{opt}}) \rangle $, is not the state with the smallest value of $\xi$.

We calculate $\xi$ for the optimal states, $| \psi_{\mathrm{opt}} \rangle $, found by the optimisation routine and observe an interesting phenomenon (see figure \ref{fig:SqueezingParameter} (a)). If the input state for the optimisation algorithm is $| \mathrm{SSS}(\mu_{\mathrm{opt}}) \rangle $, then the resultant state after the optimisation is also spin-squeezed (squares in figure \ref{fig:SqueezingParameter} (a)). In contrast, optimal states found with random input states, $| \psi_{\mathrm{rand}} \rangle $, are not necessarily spin-squeezed (dots in figure \ref{fig:SqueezingParameter} (a)). In particular, for $N > 50$, we found no such state which was spin-squeezed. We therefore conclude that, within $\mathcal{S}_N$,  spin-squeezing is not a necessary property for optimal frequency estimation.  
We remark that we qualitatively found the same result for a more general inequality for spin squeezing along any direction~\cite{korbicz_spin_2005}.  This is to be expected since squeezing in any direction can be decomposed into squeezing component along the mean spin direction and a squeezing component in an orthogonal direction.  As squeezing along the mean spin direction does not enhance the states performance for frequency estimation it follows that squeezing in any direction cannot provide better precision that squeezing in an orthogonal direction only.

We now turn to the second property, the variance with respect to the generator $h$. 
Whereas we find a large deviation of $\xi^2$ for the optimal states, one observes the contrary for $\mathcal{V}(h)$, which is the figure of merit in the noiseless case. In figure \ref{fig:MeanVar_locphase} (b), the lower blue line represents the average variance of $h$ for the optimal states $| \psi_{\mathrm{opt}} \rangle $. The standard deviation from this average is plotted as well for each $N$. However, the relative divergences from the average values are of the order of $10^{-6}$ to $10^{-3}$ and therefore hardly visible. Even though the random states used for the initialisation of the optimisation algorithm exhibit very different variances (scaling as $O(N^{2})$), the optimal states show a reduced variance that is very well approximated by $\mathcal{V}(h)=0.196 N^{1.674}$. This gives strong evidence that the variance is a crucial parameter even in noisy metrology. More importantly, the variance does not scale with $N^2$, as it does for optimal states in the noiseless case. 

\begin{figure}[htbp]\begin{picture}(420,125) \put(0,0){\includegraphics[width=.53\columnwidth]{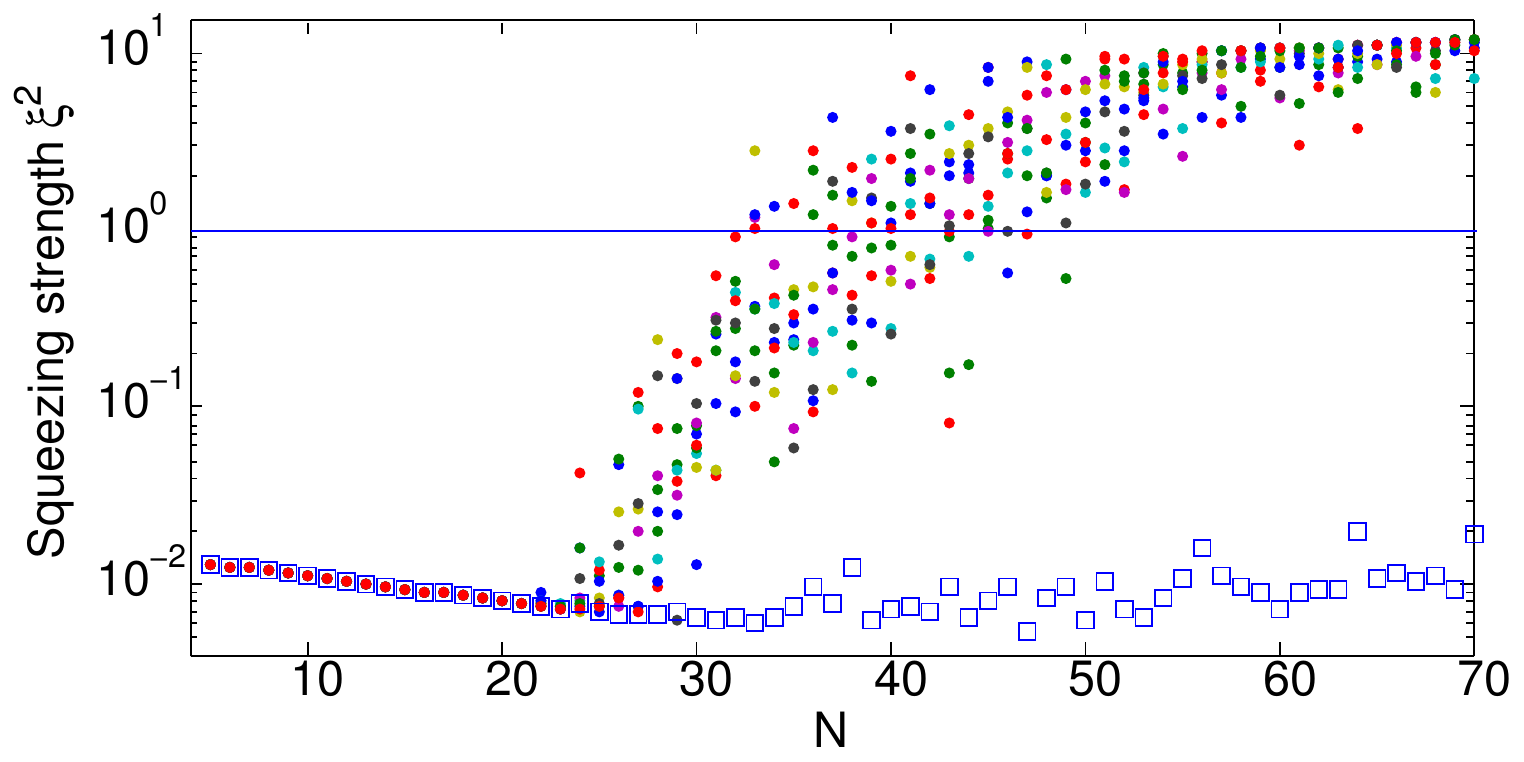}}    \put(230,-3){\includegraphics[width=.47\columnwidth]{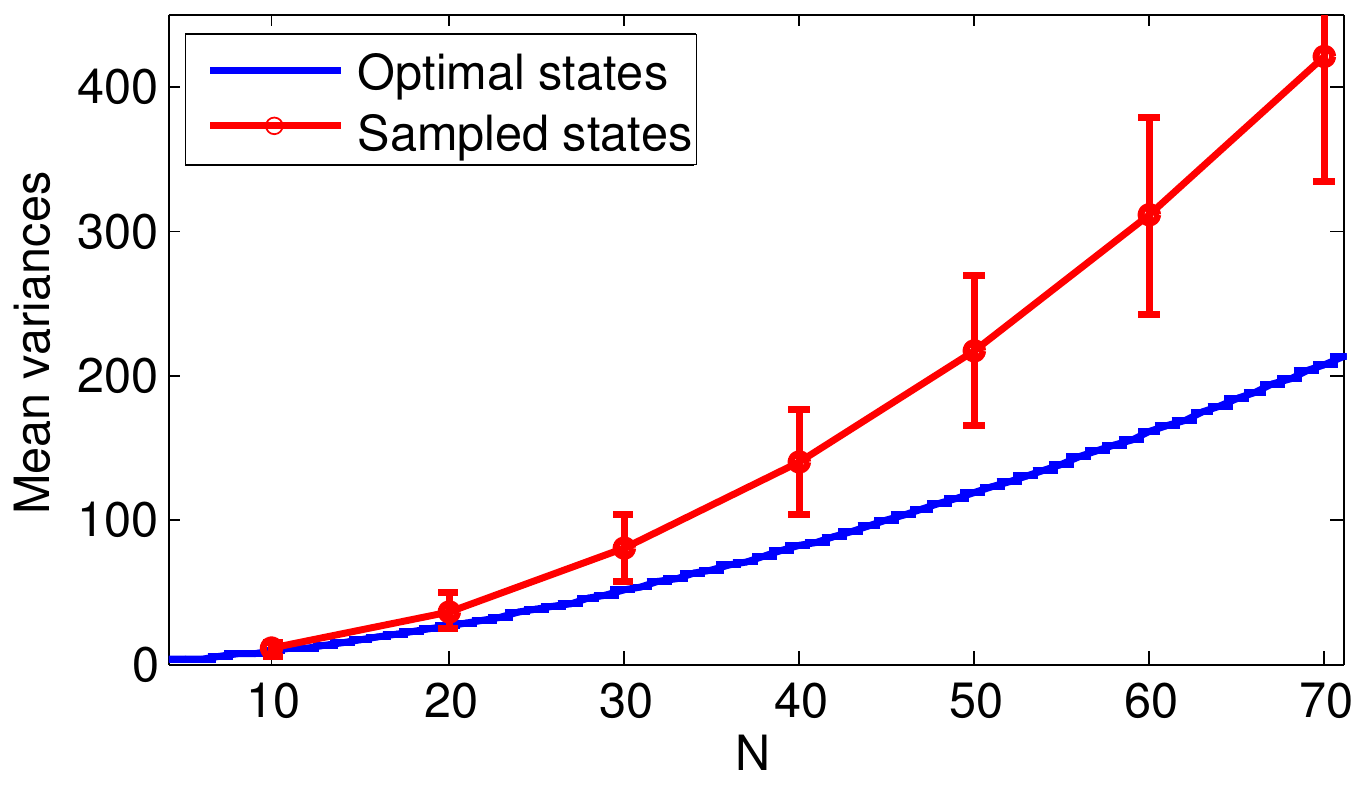}} 
    \put(0,120){(a)}
\put(239,120){(b)}
\end{picture}
\caption[]{\label{fig:SqueezingParameter} (a) Squeezing strength $\xi^2$ as defined in equation (\ref{eq:6}) for $| \psi_{\mathrm{opt}} \rangle $. The blue line represents the border below which states are spin-squeezed. In case $| \mathrm{SSS}(\mu_{\mathrm{opt}}) \rangle $ was used as an initial state in the search algorithm, the resulting state, after optimisation, was found to be spin-squeezed as well (squares). However, if the initialisation was done with a random initial state, $| \psi_{\mathrm{rand}} \rangle $, this is, for larger $N$, typically not the case.\label{fig:MeanVar_locphase} (b) Variance of $h$ for two sets of states is compared for different $N$.  The upper, red curve is the average variance of the uniformly sampled states from the ansatz class $\mathcal{S}_N$. For each $N \in \left\{ 10,20,\dots,70 \right\}$, $10^6$ states are sampled. The average value is very well fitted by $N(N+2)/12$, which is the variance of the uniform distribution. The lower, blue curve is the average variance from the eleven optimal states found with random starts (section~\ref{sec:spin-squeezing-not}). Numerically, we find that the variance is excellently fitted by $0.196 N^{1.674}$. The standard deviation (bars) from the average for these states is comparatively small.}
\end{figure}

Clearly, the variance does not fully determine the usefulness of a quantum state for noisy frequency estimation. In fact, a state that exhibits the same variance as the optimal states does not necessarily have to have an optimal QFI. For instance, the Dicke states $\left| D_1 \right\rangle $ and $\left| D_2 \right\rangle $ in figure \ref{fig:VarVsQFImax_N10_locphase} (a) are close to the optimal variance, yet one finds $\mathfrak{F}_{D_1}, \mathfrak{F}_{D_2} < \mathfrak{F}_{\mathrm{opt}}$. To understand this observation, we study the distributions of the coefficients, $c_m$, of $| \mathrm{SSS}(\mu_{\mathrm{opt}}) \rangle $ and Dicke states with similar variance. An example for $N = 75$ is shown in figure \ref{fig:Coeff_SSS_vs_Dicke} (a), where the $c_m$ of $| \mathrm{SSS}(\mu_{\mathrm{opt}}) \rangle $ and $\left| D_6 \right\rangle $ are compared. Both states have almost the same value for $\mathcal{V}(h)$. There are however two important differences: first, 
the distribution, $\{c_m\}$, of $| D_6 \rangle $ is not a ``smooth'' function of $m$ as it oscillates having $k + 1 = 7$ maxima; second, the tails of the distribution, i.e.~$c_{\pm k}$ for $k\approx N/2$, of $| D_6 \rangle $ are exponentially suppressed. We remark that $\mathfrak{F}_{\mathrm{SSS}(\mu_{\mathrm{opt}})} \approx 1.21 \mathfrak{F}_{D_6}$ in this example.

In order to discover which of the two differences above has the largest impact on $\mathfrak{F}$ we perform the following numerical test.  
We sample a large number of random states ($10^4$ for each $N \in \{10,20,\dots,70\}$). Then, the $c_m$ of the random states are multiplied by $\cos(\vartheta m)$, with $\vartheta$ chosen such that the variance of the random state is equal to the variance of the optimal states~\footnote{In the very unlikely case where no such $\vartheta$ exists, the state is discarded.}. Then, $\mathfrak{F}$ is calculated. We find that for increasing $N$, this class of modified random states typically performs much better than the unmodified random states (see section \ref{sec:typic-ansatz-stat} and figure \ref{fig:StateFamilies_locphase} (b)). What is more, the relative difference between $\mathfrak{F}_{\mathrm{opt}}$ and $\mathfrak{F}$ of the modified random states decreases with $N$, in contrast to the native random states. This indicates that, for larger $N$, the detailed distribution of $c_m$ is not so important, as these modified random states are very different on the level of single coefficients. In particular, $\mathfrak{F}$ apparently does not depend on a specific kind of ``smoothness'', which is present, e.g., for $| \mathrm{SSS}(\mu_{\mathrm{opt}}) \rangle $ (compare also to the sharpness quantity in~\cite{berry_optimal_2000,combes_states_2005}). However, the global structure of a relatively broad distribution that is suppressed at the boundary seems to play a major role.

Finally, we numerically determined that the optimal interrogation time for all optimal states can be very well approximated as $\gamma t_{\mathrm{opt}} \approx 0.4885 N^{-0.2611}$. In particular, this implies that the optimal interrogation time vanishes asymptotically, albeit much slower than for the GHZ state ($\gamma t_{\mathrm{GHZ}} = 1/(2N)$).

\subsection{Typical random states $| \psi_{\mathrm{rand}} \rangle $ perform better than product states}
\label{sec:typic-ansatz-stat}

As a side result, we remark that random states, $| \psi_{\mathrm{rand}} \rangle $, typically perform better than product states.
For each $N \in \left\{ 10,20,\dots, 70 \right\}$, we generate $10^6$ random states as described in section~\ref{sec:ansatz-space}, and for each state $\mathfrak{F}$ and $\mathcal{V}(h)$ are calculated. We find that for larger $N$ almost all random states perform better than both product states and the GHZ state (see figure \ref{fig:StateFamilies_locphase} (b)), while their mean variance is about $N(N+2)/6$ (see figure \ref{fig:MeanVar_locphase} (b)). 
This highlights the importance of $\mathcal{S}_N$, since this result is not expected for states randomly chosen from the entire Hilbert space. It also shows that a quantum improvement is nothing extraordinary and can be achieved by many different states (within $\mathcal{S}_N$).

\begin{figure}[htbp]
  \begin{picture}(420,143) \put(0,0){\includegraphics[width=.47\columnwidth]{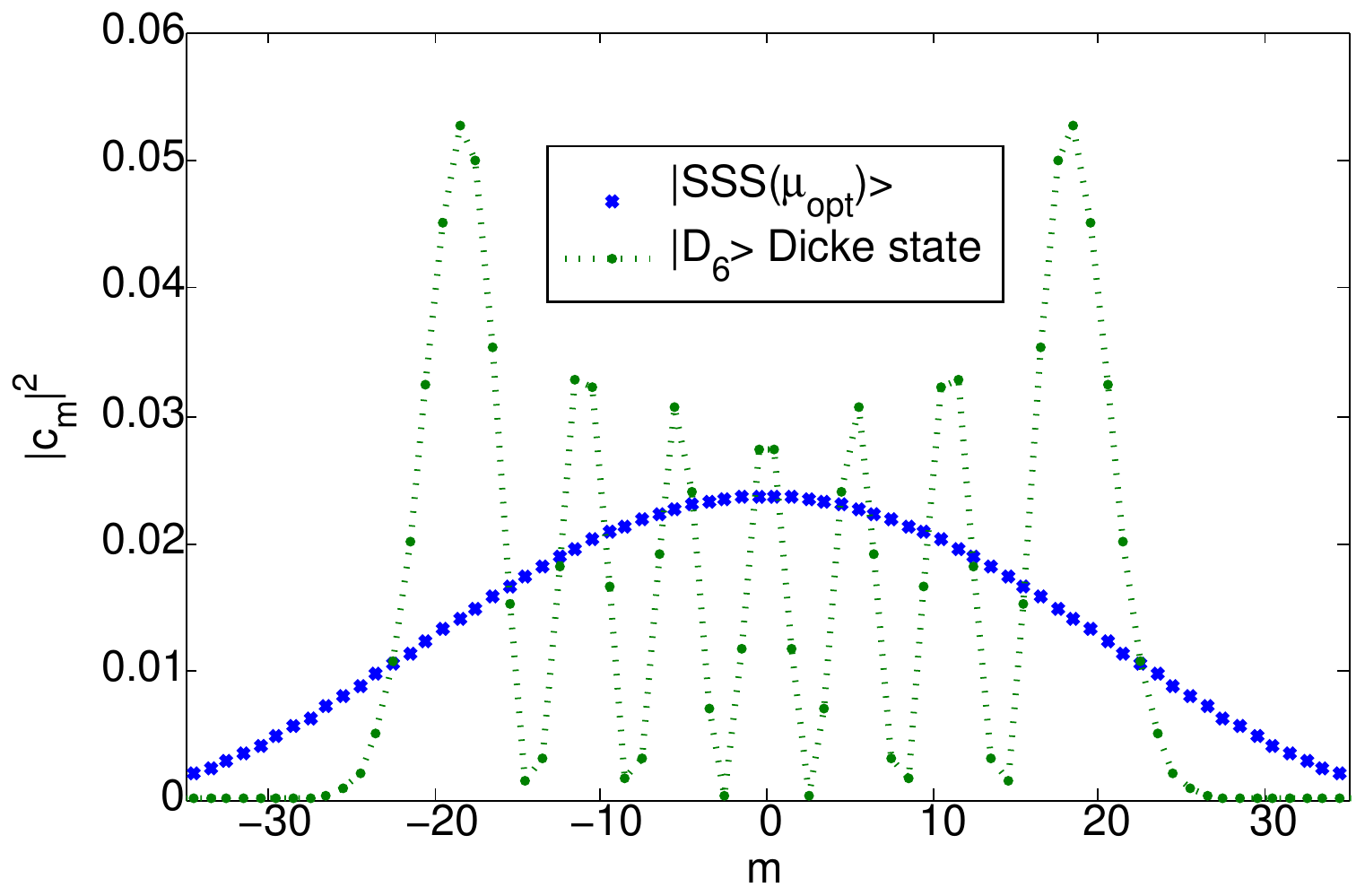}}
   \put(220,0){\includegraphics[width=.49\columnwidth]{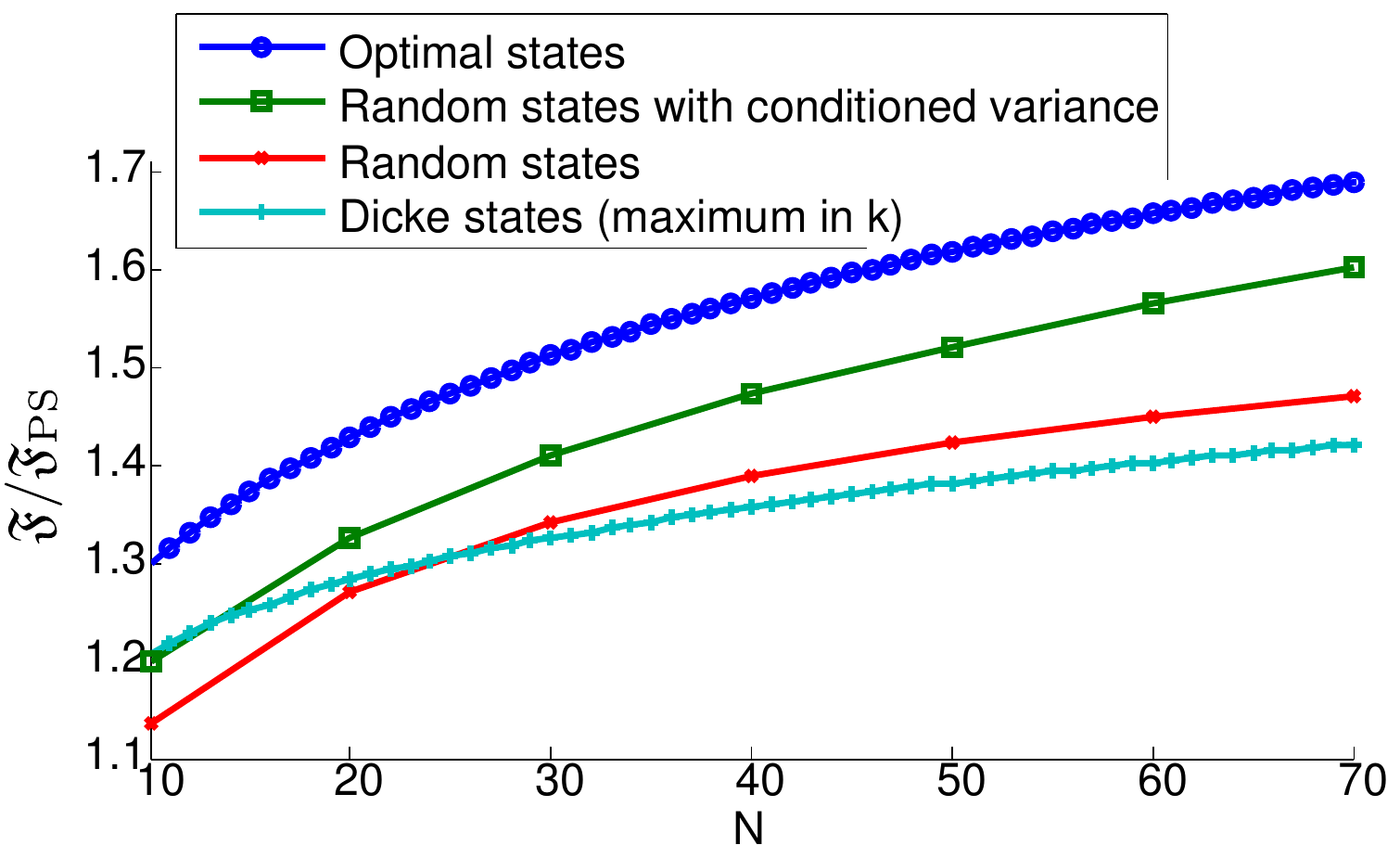}} 
    \put(0,140){(a)}
\put(220,140){(b)}
\end{picture}
\caption[]{\label{fig:StateFamilies_locphase}\label{fig:Coeff_SSS_vs_Dicke}(a) Comparison of the distributions, $|c_m|^2$, for  $| \mathrm{SSS}(\mu_{\mathrm{opt}}) \rangle $ ($\mu_{\mathrm{opt}} \approx 9.51\times 10^{-2}$) and the Dicke state $| D_6 \rangle $ for $N = 75$. Both states exhibit almost the same value for $\mathcal{V}(h)$. Apart from the oscillations in the distribution of $| D_6 \rangle $, the main difference is that for $| D_6 \rangle $ the coefficients close to $\pm N/2$ are exponentially suppressed. (b) Maximal QFI, $\mathfrak{F}$, relative to $\mathfrak{F}_{\mathrm{PS}}$ (and $\mathfrak{F}_{\mathrm{GHZ}}$) for different state classes. The results are shown for Dicke states $| D_k \rangle $ (equation \eqref{Dickestates}) where we optimise over $k$; for the optimal states $| \psi_{\mathrm{opt}} \rangle $ (where the values of $\mathfrak{F}_{\mathrm{opt}}$ and $\mathfrak{F}_{\mathrm{SSS}(\mu_{\mathrm{opt}})}$ are in this plot indistinguishable); and for the unconditioned as well as the conditioned random states, $| \psi_{\mathrm{rand}} \rangle $. For $N = 70$, the maximal improvement is about $1.69 \mathfrak{F}_{\mathrm{PS}}$, which corresponds to a minimal uncertainty of about $0.76 (\delta \omega)_{\mathrm{PS}}$. The standard deviations of the results for the random states are indicated as shaded areas of the same colour. }
\end{figure}

\section{Beyond the standard scenario}
\label{sec:beyond-stand-scen}

The scenario with local Hamiltonian and dephasing noise discussed in section~\ref{sec:standard-scenario} is certainly an important instance in frequency estimation theory. However, many experimental setups have to be described by different unitary evolutions and/or different noise models. Hence, the investigation of scenarios different than the standard scenario is of fundamental and practical relevance. This section is devoted to two such scenarios. In section~\ref{sec:local-depolarization} we consider the situation where local dephasing noise is replaced by local depolarisation noise (equation \eqref{localnoise} with $\mu_x=\mu_y=\mu_z=1/3$), which can be viewed as a combination of local dephasing and transversal bit-flip noise, and show that this change does not lead to a dramatically different picture. In section~\ref{sec:correlated-noise} we consider spatial and temporal correlated dephasing noise, described by Lindblad operator given in equation \eqref{eq:53} and show that the GHZ state remains superior to product states.
 In fact, a favourable scaling of the QFI can be achieved for the GHZ state.

\subsection{Local depolarisation}
\label{sec:local-depolarization}
\begin{comment}
In~\cite{chaves_noisy_2013}, it was shown that the SQL can be overcome if dephasing noise is replaced by transversal bit-flip noise. There, it was shown that $\mathfrak{F}$ reaches a scaling of $O(N^{5/3})$ for the GHZ state, which is close to Heisenberg scaling. Indeed, we recently showed that the Heisenberg limit can be achieved in this scenario by using error correction techniques~\cite{dur_improved_2013}. It is natural to ask what happens if there are deviations from this perfectly directed noise if one does not make use of error correction techniques. The answer was already given in~\cite{chaves_noisy_2013}, where it was shown that in the asymptotic limit one obtains again the SQL. However, it is unknown which states are optimal in such a scenario.
\end{comment}
In this section, we study the local depolarisation noise, equation \eqref{localnoise} with $\mu_x=\mu_y=\mu_z=1/3$. Physically, this model describes uncorrelated interactions between the system qubits and a bath at infinite temperature. As can be easily seen, the minimal error achievable with product states as initial states is the same as for dephasing noise, i.e.~$(\delta \omega)_{\mathrm{PS}} = \sqrt{2 \gamma e/(N T)}$. Recently, an upper bound on $\mathfrak{F}$ was obtained~\cite{demkowicz-dobrzanski_elusive_2012}, stating that $\mathfrak{F}_{\mathrm{opt}} \leq 4 e/3 \mathfrak{F}_{\mathrm{PS}}$. Thus the minimal uncertainty $\delta \omega$ for any initial state scales with the SQL and is at most $0.53 (\delta \omega)_{\mathrm{PS}}$, which is smaller than for local dephasing noise. It is not known, however, whether this bound is tight. As mentioned before, the one axis-twisted SSS were shown to asymptotically reach the (optimal) minimal error. This was done by calculating a lower bound $G\leq\mathfrak{F}$ (see \ref{sec:asymptotic-behavior}) that reaches the ultimate precision limit in the standard scenario. As the replacement of phase noise by 
depolarisation noise does not change $G$, it follows that, asymptotically, $| \mathrm{SSS}(\mu_{\mathrm{opt}}) \rangle $ reaches at least the same precision with depolarisation noise. To our knowledge there exists no example of a state that overcomes this limit. Therefore, it is not clear whether the bound of \cite{demkowicz-dobrzanski_elusive_2012} can actually be achieved.

Let us now come to the results of our numerical studies. The algorithm to compute the QFI in case of this noise model is still efficient for symmetric ansatz states of equation \eqref{eq:5}. However, in practice it is much more demanding (see \ref{sec:numerical-methods}) and we are therefore limited to system sizes up to $N = 30$ (see figure \ref{fig:NvsQFImax_whitenoise}). Whereas for $N = 2,3$, the GHZ state is the optimal initial state, this is no longer the case for $N \geq 4$. However, the GHZ state remains superior to the product state. The QFI for the GHZ state reads (see \ref{sec:maximal-qfi-ghz})
\begin{equation}
\label{eq:21}
\mathcal{F}_{\mathrm{GHZ}} = \frac{2^N e^{-2 N \gamma t} \, t^2}{\left( 1+e^{-\gamma t} \right)^N +\left( 1-e^{-\gamma t} \right)^N}.
\end{equation}
For large $N$, the maximisation of $\mathcal{F}/t$ over time leads to $\mathfrak{F}_{\mathrm{GHZ}} = 2N T/(3 \gamma e)  \approx 1.33 \mathfrak{F}_{\mathrm{PS}}$. We remark that the superior performance of the GHZ state in this scenario can be understood from the results of~\cite{Chaves:13} for the case where the noise has components parallel and perpendicular to the direction of the unitary.  Note that this improvement is already achieved for a relatively small $N$. The optimal states are found using the algorithm described in section~\ref{sec:spin-squeezing-not}. One observes that the relative difference between the optimal states and $| \mathrm{SSS}(\mu_{\mathrm{opt}}) \rangle $ is larger than in the standard scenario. For $N\leq 30$, it is in the range of a few percent. However, this gap seems to vanish as $N$ increases such that one could expect that SSS become optimal for larger $N$ (see figure \ref{fig:NvsQFImax_whitenoise}).

\begin{figure}[htbp]
\centerline{\includegraphics[width=\columnwidth]{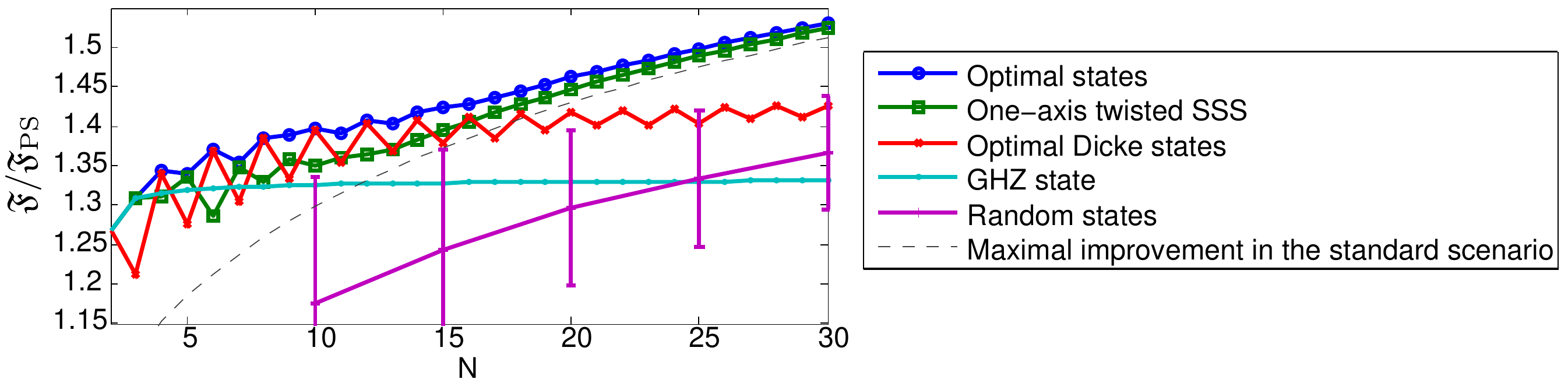}}
\caption[]{\label{fig:NvsQFImax_whitenoise} Maximal QFI for the different state families discussed in section~\ref{sec:ansatz-space} in the presence of depolarisation noise. For comparison, the maximal relative improvement from the standard scenario is plotted as well. It seems that for larger $N$, the difference in the maximal improvement between both scenarios vanishes. For $N = 30$, the minimal achievable uncertainty is $1.53 \mathfrak{F}_{\mathrm{PS}}$, which corresponds to $0.81 (\delta \omega)_{\mathrm{PS}}$.}
\end{figure}

Compared to the standard scenario of section~\ref{sec:standard-scenario}, the best Dicke states are closer to the optimal states for small $N$. However, the difference in their performance as compared to the standard scenario seems to vanish for larger $N$ (see figures \ref{fig:StateFamilies_locphase} (b) and~\ref{fig:NvsQFImax_whitenoise}). Also the average improvement of the randomly chosen initial states ($10^6$ random states for each $N \in \left\{ 10,15,\dots, 30 \right\}$) is compatible with that of the phase noise scenario. Altogether, these results are qualitatively very similar to those obtained in the standard scenario. It is worth noting, however, that there are certain differences in the asymptotic case, e.g.~that the GHZ state is superior to the product state in case of depolarising noise.

\subsection{\label{sec:correlated-noise}Correlated dephasing}

We now discuss our findings for the case of correlated dephasing noise, a dominant source of noise in experiments based on trapped ions~\cite{kielpinski_decoherence-free_2001,langer_long-lived_2005,roos_designer_2006,monz_14-qubit_2011} and on the stability of atomic clocks \cite{buzek_optimal_1999,andre_stability_2004,borregaard_near-heisenberg-limited_2013,kessler_heisenberg-limited_2013} that arises mainly due to fluctuations in the phase reference that collectively affect all the ions. It was shown that optimal states for the standard scenario offer no quantum advantage in this case as spatially correlated noise causes quantum coherences to vanish even faster than in the case of local dephasing noise~\cite{dorner_quantum_2012}. However, clever use of decoherence-free subspaces~\cite{dorner_quantum_2012}, or higher-dimensional probe systems, i.e.~qutrits~\cite{jeske:2013}, allows for higher precision compared to uncorrelated noise and can even restore Heisenberg scaling in precision. On the other hand, temporal correlations can lead to an improved performance of highly entangled states such as the GHZ state~\cite{Chin:12, matsuzaki_magnetic_2011}.  
The reason is that the optimal measurement time $t$ decreases with $N$, eventually entering regime where $\gamma t \ll 1$. Then, (short) temporal correlations lead to a quadratic suppression of the noise strength.

Here, we will focus on qubits and study the effect of temporal as well as spatial correlated noise as defined in equation (\ref{eq:53}). 
For spatially correlated noise, the master equation reads as in equation \eqref{eq:2}. An additional (time-dependent) factor multiplying the super-operator $\mathcal{L}$ (see equation (\ref{eq:53})) models the time-evolution assumed in~\cite{monz_14-qubit_2011}. For $\gamma t \gg 1$, we are in the Markovian regime $ f(t) \approx 1 $, while for $\gamma t \ll 1$, the effect of noise is suppressed through the temporal correlations by $ f(t) \approx \gamma t$.

For the GHZ state, pure spatial correlations lead to $\mathfrak{F}$ constant in $N$~\cite{dorner_quantum_2012}, whereas pure temporal correlations give rise to an improved scaling of the optimal QFI of $\mathcal{O}(N^{3/2})$~\cite{Chin:12}. With $h = S_z$, the QFI for the GHZ states with both kinds of correlations equals
\begin{equation}
\label{eq:16}
\mathcal{F}_{\mathrm{GHZ}} = t^2 N^2 \exp \left[ -2 \left( \gamma t + e^{-\gamma t} - 1 \right) N^2 \right].
\end{equation}
For $N\gg 1$ the optimisation over time can be easily performed and leads to $\mathfrak{F}_{\mathrm{GHZ}} = NT/(\sqrt{2 e} \gamma)$, i.e.~the SQL.

We now compare this result to the one obtained for product states as initial states. Note that, in order to avoid the effect of the correlated noise, one might use a sequential approach, where one qubit after the other is exposed to the evolution. In this case the total time is given by $NT$, where $T$ denotes the evolution time for a single experiment. Then, $\mathfrak{F}_{\mathrm{PS}}$ coincides with the one obtained for local dephasing. The improvement of using a GHZ state is hence limited to a constant factor $\mathfrak{F}_{\mathrm{GHZ}} \approx 2.37 \mathfrak{F}_{\mathrm{PS}}$. If, however, a parallel scheme is employed, where a $N$-qubit product state is used as initial state, the effect of noise changes the scaling of $\mathfrak{F}$. Indeed, numerically we find that then $\mathfrak{F}_{\mathrm{PS}} = O(N^{1/2})$, which is worse than the SQL. Hence, for a parallel scheme, we indeed encounter a quantum improvement in scaling. This is because, for the GHZ state, the time-like and space-like 
correlations partially compensate each other. In contrast, the product state, whose optimal measurement time is relatively long (and constant in $N$), cannot benefit from short correlations in time as much as it is impaired by spatial correlations. It is worth mentioning that, in this scenario, the equivalence between the sequential and parallel strategies for product states, discussed in~\cite{maccone_intuitive_2013}, is not valid.

Numerically we searched for the optimal states within the ansatz class of equation \eqref{eq:5} for $N \leq 75$ using the optimisation algorithm \cite{nelder_simplex_1965} as explained in section \ref{sec:ansatz-space}. We find that the GHZ states are the optimal states. We believe that this is due to the fact that the GHZ state has the shortest optimal measurement time and, for large enough $N$, its entire time evolution occurs in the regime where the strength of the noise is quadratically suppressed.

\section{Summary and conclusions}
\label{sec:summary-conclusions}

In this paper, we investigated different frequency estimation scenarios mainly by numerical means.

In the scenario of local dephasing, we numerically identified optimal states within a system size up to 70 qubits. We found that one-axis twisted SSS are almost optimal for a certain squeezing parameter. However, we also gave strong evidence that the property of being spin-squeezed is not necessary by presenting optimal states that are not spin-squeezed. Moreover, we identified a key property all optimal states have in common, namely the variance with respect to the Hamiltonian as a consequence of a specific global distribution of the coefficients $c_{m}$ in equation \eqref{eq:5}. Interestingly, the variance of the optimal states scales as $N^p$ with $p$ strictly smaller than two.

Furthermore, we showed that similar results apply to the case of local depolarisation noise. However, there the GHZ state outperforms product states. In contrast to that, we demonstrated that more drastic changes of the dynamics give qualitative different results. In the case of temporal and spatial correlation of noise, the scaling of $\mathfrak{F}$ was shown to be better for the GHZ state than for the product state. Note that recently phase estimation in the presence of local or spatially correlated noise has been investigated, and a similar family of states as presented here have been identified as optimal states for parameter estimation for large $N$~\cite{Knysh:14}.

The asymptotic behaviour is in general an interesting question. Concerning the standard scenario, for example, one could ask whether the gap in $\mathfrak{F}$ between the optimal states $| \psi_{\mathrm{opt}} \rangle $ and other state families like Dicke states or random states persists or vanishes in the limit of large $N$. In \ref{sec:asymptotic-behavior}, a lower bound on $\mathfrak{F}$ is investigated. While, for specific SSS, this bound asymptotically reaches the maximal value $\max_{N \rightarrow \infty}\mathfrak{F}_{\mathrm{SSS}(\mu_{\mathrm{opt}})} = NT/ (2\gamma)$, we analytically find that, for a state family with $\mathcal{V}(h) = O(N^2)$, this quantity is strictly below $ NT/ (2\gamma)$. This underlines the importance of the variance (or more generally, the global spectral distribution) and could indicate that a too large variance prohibits optimal performance.

Beyond this, there are further big open questions.  For example, why is there asymptotically no improvement in scaling for the standard scenario? Clearly, $\mathfrak{F}$ scales with $N^{g(N)}$, where $g(N) > 1$ for finite $N$ but is strictly one for $N \rightarrow \infty$. Furthermore, the numerical findings indicate which variance $\mathcal{V}(h)$ is optimal, but it is not clear why it has this particular value. 

From a more practical point of view our results show that whilst one-axis twisted SSS are nearly optimal, there are a plethora of other states from $\mathcal{S}_N$ that perform equally well, or better than the one-axis twisted SSS. Thus, if spin-squeezing is an easily-implementable operation in a given practical realization of frequency estimation, it is a good resource for frequency estimation.  Moreover, for such states a comparatively simple measurement is known to be asymptotically optimal \cite{ulam-orgikh_spin_2001}. However, in experimental realizations where spin-squeezing is not available, or is an expensive operation to implement, our results show that optimal precision might still be possible as there exist nonsqueezed states within $\mathcal{S}_N$ that will be just as good as one-axis twisted SSS.  Hence, our findings open the way to identify high performance states for frequency estimation that may be easily prepared in certain set-ups.  Which states are easiest to prepare, and at what cost, for a certain experimental set-up is an interesting question that warrants further investigation. Moreover,  we showed that for ion trap experiments with spatially and temporally correlated noise,  the GHZ state is not only superior to product states used in parallel, but also achieves a constant improvement even compared to the sequential strategy with single ions.  

An open problem that is also relevant from the experimental point of view is the choice the optimal measurement.
By calculating the QFI, the optimisation of the measurement is intrinsically performed. In fact, for states other than SSS, where it is known that the optimal measurement is local, nothing substantial is known about the optimal measurement basis (apart from the fact that it is the projection onto the eigenspaces of the symmetric logarithmic derivative).

\ack
This work was supported by the Austrian Science Fund (FWF): Grant numbers P24273-N16, Y535-N16, SFB F40-FoQus F4012-N16, and J3462.

\appendix

\section{Details on the numerical method}
\label{sec:numerical-methods}

Here, we present some details of the numerical methods for an efficient calculation of the QFI in the subspace of permutationally invariant density matrices.

We first treat the case of local phase noise with local Hamiltonian (section~\ref{sec:standard-scenario}) and comment later on other scenarios. In the presence of local phase noise, which is described by $\mathcal{E}^{\otimes N}$, where $\mathcal{E}(\rho)=e^{-\gamma t}\rho +(1-e^{-\gamma t}) \sigma_z \rho \sigma_z$, the density matrix $\rho = \mathcal{E}^{\otimes N}(\left| \psi \right\rangle\!\left\langle \psi\right| )$ is no longer permutationally symmetric, but it is still permutationally invariant. Therefore, there exists an efficient representation of $\rho$ and, as we will see, also an efficient way to calculate the eigenvalues and eigenstates, even though the rank of $\rho$ is in general exponential in $N$.

In this appendix, we use the common orthonormal eigenbasis of $S^2 = S_x^2 + S_y^2 + S_z^2$ and $S_z$, which we denote by $| j,m,\alpha \rangle $. One has that $S^2 \left| j,m,\alpha \right\rangle = j(j+1) \left| j,m,\alpha \right\rangle $ and $S_z \left| j,m,\alpha \right\rangle = m \left| j,m,\alpha \right\rangle $ for any $\alpha$, which labels the degeneracies. All the states in the ansatz class (see equation (\ref{eq:5})) satisfy $j = j_{\mathrm{max}} = N/2$. As there is no degeneracy in the common eigenstates of $S^2$ and $S_z$ for $j = j_{\mathrm{max}}$ we skip the index $\alpha$ in the following.

\subsection{Permutationally invariant operator basis}
\label{sec:perm-invar-oper}
In this sub-appendix we introduce the permutationally invariant operator basis of Hartmann~\cite{hartmann_generalized_2012} and determine the action of these operators on the basis states $\{\ket{j,m,\alpha}\}$. 

Let $\mathcal{H}^{\otimes N}$ be the state space of $N$ spin-$1/2$ systems.  Defining $\ket{0}\equiv\ket{1/2,1/2},\, \ket{1}\equiv\ket{1/2,-1/2}$, any $\rho\in\mathcal{B}(\mathcal{H})$ can be written as
\begin{equation}
\rho=\sum_{i_1,i_2,i_3,i_4} \rho_{(i_1,i_2,i_3,i_4)} X_{(i_1,i_2,i_3,i_4)}, 
\label{operatorbasis}
\end{equation}
where 
\begin{equation}
\label{eq:15}
X_{(i_1,i_2,i_3,i_4)} = \left| 0
  \right\rangle\!\left\langle 0\right| ^{\otimes i_1} \otimes \left|0
  \right\rangle\!\left\langle 1\right| ^{\otimes i_2}\otimes \left|1
  \right\rangle\!\left\langle 0\right| ^{\otimes i_3}\otimes \left|1
  \right\rangle\!\left\langle 1\right| ^{\otimes i_4},
\end{equation}
form a basis of $\mathcal{B}(\mathcal{H})$ and the sum in equation \eqref{operatorbasis} runs over all $i_1,i_2,i_3,i_4$ satisfying $\sum_{j=1}^4 i_j=N$.

As $\{\ket{j,m,\alpha}\}$ is an orthonormal basis, the operators of equation \eqref{eq:15} can be written as
\begin{equation}
X_{(i_1,i_2,i_3,i_4)}=\sum_{j,j^{\prime},\alpha, \alpha^{\prime}} c^{(j,j^{\prime},\alpha, \alpha^{\prime})}_{m,m^{\prime}, d}\ketbra{j,m,\alpha}{j^{\prime}, m^{\prime},\alpha^{\prime}},
\label{eq:200}
\end{equation}
where $m=i_1+i_2-N/2, \,m^{\prime}=i_1+i_3-N/2$, and $2d\equiv i_1+i_3$ can be interpreted as the number of places where the symbols in the corresponding bra and ket in equation \eqref{eq:15} are different~\cite{hartmann_generalized_2012} .  Let $\pi: S_N\to\mathcal{H}^{\otimes N}$ be a unitary representation of the permutation group of $N$ objects, $S_N$. Using Schur's lemmas~\cite{Chen:89}, the unitary representation $\{\pi_g;\, g\in S_N\}$ can be decomposed into irreducible representations of $S_N$.  Indeed, in the $\{\ket{j,m,\alpha}\}$ basis all $\{\pi_g;\, g\in S_N\}$ take the form
\begin{equation}
\pi_g=\sum_{j} \pi^{(j)}_g\otimes \one_{d_j}.
\label{eq:201}
\end{equation}
Here, $j$ labels the inequivalent irreducible representations on $S_N$, $\pi^{(j)}$, present in the unitary representation $\pi$, and $d_j$ denotes the multiplicity of the corresponding $\pi^{(j)}$. Consequently, the total Hilbert space, $\mathcal{H}^{\otimes N}$, can be conveniently written as 
\begin{equation}
\mathcal{H}^{\otimes N}\cong\bigoplus_j \cM^{(j)}\otimes \cN^{(j)},
\label{eq:202}
\end{equation}
where the space $\cM^{(j)}\equiv\mathrm{span}\{\ket{j,\alpha}_{\alpha=1}^{\Delta_j}\}$ of dimension, $\Delta_j\equiv\mathrm{dim}(\pi^{(j)})=\binom{N}{N/2-j}\frac{2j+1}{N/2+j+1}$, is the space upon which the irreducible representation $\pi^{(j)}$ acts non-trivially, and the space $\cN^{(j)}\equiv\mathrm{span}\{\ket{j,m}_{m=1}^{d_j}\}$, of dimension $d_j=2j+1$, is the space upon which $\{\pi_g; g\in S_N\}$ act trivially~\cite{BRS07}.  The tensor product in equation \eqref{eq:202} is not a tensor product between real physical systems, but rather a tensor product between two virtual systems described by the state spaces $\cM^{(j)}$ and $\cN^{(j)}$ respectively. The states of these virtual systems can be defined via the isomorphism $\ket{j,m,\alpha}\equiv\ket{j,\alpha}\otimes\ket{j,m}$.      

Hartmann showed that an operator basis for the permutationally symmetric states is given by~\cite{hartmann_generalized_2012}    
\begin{equation}
\label{eq:27}
K_{m,m^{\prime},d} \mathrel{\mathop:}= \frac{1}{i_1! i_2! i_3! i_4!}\sum_{g\in S_N} \pi_g\,(X_{(i_1,i_2,i_3,i_4)})\pi_g^{\dagger}.
\end{equation}
Note that the definition of the basis operators in equation (\ref{eq:27}) differ from that of \cite{hartmann_generalized_2012} by a factor $N!/(i_1! i_2! i_3! i_4!)$. Using Schur's lemmas it can be shown that equation \eqref{eq:27} can be written with respect to the $\{\ket{j,m,\alpha}\}$ basis as 
\begin{equation}
K_{m,m^{\prime},d}=\frac{N!}{i_1! i_2! i_3! i_4!} \sum_j\left(\cD_{\cM^{(j)}}\otimes\cI_{\cN^{(j)}}\right)\circ \cP^{(j)}\left[X_{(i_1,i_2,i_3,i_4)}\right],
\label{eq:203}
\end{equation}
where $\cD$ is the completely depolarising map, $\cD[A]=\frac{\mathrm{tr}(A)}{\mathrm{dim}(\cH)}\one, \, \forall A\in\cB(\cH)$, $\cI$ is the identity map, and $\cP^{(j)}[A]=\Pi_j A \Pi_j$, where $\Pi_j$ is the projector onto the space $\cM^{(j)}\otimes\cN^{(j)}$~\cite{BRS07}. Careful counting gives a total number of $1/6(N+1)(N+2)(N+3)$ different basis operators for the space of permutationally invariant matrices~\cite{hartmann_generalized_2012}.

Plugging equation \eqref{eq:200} into equation \eqref{eq:203} yields
\begin{align}\nonumber
K_{m,m^{\prime},d}&=\frac{N!}{i_1! i_2! i_3! i_4!}\sum_j\sum_{\alpha,\alpha^{\prime}}\left(\cD_{\cM^{(j)}}\otimes\cI_{\cN^{(j)}}\right)\left[c^{(j,\alpha,\alpha^{\prime})}_{m,m^{\prime},d}\ketbra{j,\alpha}{j,\alpha^{\prime}}\otimes\ketbra{j,m}{j,m^{\prime}}\right]\\
&=\frac{N!}{i_1! i_2! i_3! i_4!}\sum_j\frac{\sum_{\alpha=1}^{\Delta_j} c^{(j,\alpha,\alpha)}_{m,m^{\prime},d}}{\Delta_j}\one_{\cM^{(j)}}\otimes\ketbra{j,m}{j,m^{\prime}}.
\label{eq:204}
\end{align}
Defining $\frac{N!}{i_1! i_2! i_3! i_4!}\frac{\sum_{\alpha=1}^{\Delta_j} c^{(j,\alpha,\alpha)}_{(m,m^{\prime},d)}}{\Delta_j}=\mu^{(j)}_{m,m^{\prime},d}$, and noting that $\one_{\cM^{(j)}}=\sum_{\alpha=1}^{\Delta_j}\ketbra{j,\alpha}{j,\alpha}$, equation \eqref{eq:204} becomes
\begin{align}\nonumber
K_{m,m^{\prime},d}&=\sum_j\sum_{\alpha=1}^{\Delta_j} \mu^{(j)}_{m,m^{\prime},d}\ketbra{j,\alpha}{j,\alpha}\otimes\ketbra{j,m}{j,m^{\prime}}\\
&=\sum_j\sum_{\alpha=1}^{\Delta_j} \mu^{(j)}_{m,m^{\prime},d}\ketbra{j,m,\alpha}{j,m^{\prime},\alpha}
\label{eq:205}
\end{align}
where we have made use of the isomorphism $\ket{j,m,\alpha}\equiv\ket{j,\alpha}\otimes\ket{j,m}$ in the last line of equation \eqref{eq:205}.  Finally, defining $P^{(j)}_{m,m^{\prime}}\equiv\sum_{\alpha=1}^{\Delta_j}\ketbra{j,m,\alpha}{j,m^{\prime},\alpha}$ gives
\begin{equation}
    \label{eq:29}
    K_{m,m^{\prime},d} = \sum_{j = j_{\min}}^{N/2} \mu_{m,m^{\prime},d}^{(j)} P_{m,m^{\prime}}^{(j)},
\end{equation}
where $j_{\min}=0\, (1/2)$ for $N$ even (odd).

Thus, the action of the permutationally symmetric operators, equation \eqref{eq:27} on the basis states $\{\ket{j,m,\alpha}\}$ is given by equation \eqref{eq:29}.  The action of phase noise $\mathcal{E}$ is particularly easy to express in this basis. Since $\mathcal{E}(\left| 0\right\rangle\!\left\langle 0\right| ) = \left| 0\right\rangle\!\left\langle 0\right| $, $\mathcal{E}(\left| 1
\right\rangle\!\left\langle 1\right| ) = \left| 1\right\rangle\!\left\langle 1\right| $ and $\mathcal{E}(\left| 0\right\rangle\!\left\langle 1\right| ) = e^{-\gamma t }\left| 0
\right\rangle\!\left\langle 1\right| $, one finds that
\begin{equation}
\label{eq:30}
\mathcal{E}^{\otimes N} K_{m,m^{\prime}, d} = e^{-2 d \gamma t }  K_{m,m^{\prime}, d}.
\end{equation}

\subsection{Efficient representation of permutationally invariant states under phase noise}
\label{sec:effic-repr-perm}

In this sub-appendix we seek an efficient description of states from our ansatz space (see equation \eqref{eq:5}) in terms of the Hartmann operators $K_{m,m^{\prime},d}$ of equation \eqref{eq:27}.  Consider an arbitrary pure state, $\rho=\sum_{m,m^{\prime}} c_m\,c_{m^{\prime}} \ketbra{N/2,m}{N/2,m^{\prime}}$ where $c_m  \in \left[ 0,1 \right]$, belonging to our ansatz space. For $m\geq m^{\prime}$, and $m+m^{\prime}\leq 0$ the matrix elements $|N/2,m\rangle\!\langle N/2,m^{\prime}|$ can be expressed as

\begin{align}\nonumber
\sqrt{\binom{N}{N/2+m}\binom{N}{N/2+m^{\prime}}}\ketbra{N/2,m}{N/2,m^{\prime}}&=\sum_{g\in S_N}\pi_g\left(\sum_{i_1,i_2,i_3,i_4}\frac{1}{i_1! i_2! i_3! i_4!}X_{i_1,i_2,i_3,i_4}\right)\pi_g^{\dagger}\\
&=\sum_{i_1,i_2,i_3,i_4}\left(\frac{1}{i_1! i_2! i_3! i_4!}\sum_{g\in S_N}\pi_g\, X_{i_1,i_2,i_3,i_4}\,\pi_g^{\dagger}\right)
\label{eq:301}
\end{align}
where the second sum in the first line of equation \eqref{eq:301} runs over all $i_1,i_2,i_3, i_4$ satisfying $i_1+i_2=N/2+m$, $i_1+i_3=N/2+m^{\prime}$. Using equation \eqref{eq:203} it follows that the expression in brackets in the second line of equation \eqref{eq:301} is equal to $K_{m,m^{\prime},d}$.  In addition, as $m,\, m^{\prime}$ are fixed the sum over $i_1,i_2,i_3, i_4$ must yield operators  $K_{m,m^{\prime}, d}$ with different $d$'s. Recalling that $2d$ counts the number of non-coincidences between the strings in the ket and bra of $X_{i_1,i_2,i_3,i_4}$ it is easy to determine that, for $m\geq m^{\prime}$, $m+m^{\prime}\leq 0$, $\frac{m-m^{\prime}}{2}\leq d\leq \frac{N-{|m+m'|}}{2}$. Similar arguments hold for all other cases so that generally one finds
\begin{equation}
\label{eq:19}
| N/2,m\rangle\!\langle N/2,m^{\prime}| = \frac{1}{\sqrt{\binom{N}{N/2+m}\binom{N}{N/2+m^{\prime}}}} \sum_{d = \left|
      m-m^{\prime} \right|/2}^{(N-\left| m+m^{\prime} \right|)/2}
  K_{m,m^{\prime}, d}.
\end{equation}

Using equations (\ref{eq:30},~\ref{eq:19}) the action of phase noise on the state $\rho$ can be easily evaluated to be
\begin{equation}
\label{eq:31}
  \rho =\mathcal{E}^{\otimes N} \left( \ket{\psi}\!\bra{\psi}\right)=
\sum_{m,m^{\prime}} \frac{c_m c_{m^{\prime}}}{\sqrt{\binom{N}{N/2+m}\binom{N}{N/2+m^{\prime}}}} \sum_{d =  \left| m-m^{\prime} \right|/2}^{(N-\left| m+m^{\prime} \right|)/2}
  e^{-2 d \gamma t} K_{m,m^{\prime}, d}.
\end{equation}
Note that in order to compute the QFI, we do not need to apply the unitary operator, $U=e^{iHt}$, since it commutes with the Hamiltonian, i.e.~$F(\rho)=F(U\rho U^\dagger)$.  Defining $\lambda_{m,m^{\prime}}^{(j)} \mathrel{\mathop:}= \sum_{d }c_m c_{m^{\prime}}/\sqrt{\binom{N}{N/2+m}\binom{N}{N/2+m^{\prime}}} e^{-2 d \gamma t} \mu^{(j)}_{m,m^{\prime},d}$, we have
\begin{equation}
\label{eq:32}
\rho = \sum_{j,m,m^{\prime}} \lambda_{m,m^{\prime}}^{(j)} P_{m,m^{\prime}}^{(j)}.
\end{equation}
Therefore, we find an efficient description of the noisy state $\rho$ in terms of the operators $P_{m,m^{\prime}}^{(j)}$. As the states $\ket{j,m,\alpha}$ are all orthogonal, the eigenvalues of $\rho$ are obtained by diagonalising the matrices $\Lambda^{(j)} = \left\{ \lambda_{m,m^{\prime}}^{(j)} \right\}_{m,m^{\prime}}\in \mathbbm{C}^{(2j+1)\times (2j+1)}$, which can be done efficiently in $N$.

\subsection{Determination of $\mu^{(j)}_{m,m^{\prime},d}$}
\label{sec:eigen-sing-valu}

The crucial part is now to find the transition values $\mu^{(j)}_{m,m^{\prime},d}$. First, consider the case $m = m^{\prime}$, which renders the operators $K_{m,m,d}$ self-adjoint and the transition values become eigenvalues. Note that $K_{m,m,0}$ are projectors on the eigenspace of $S_z$ with eigenvalue $m$, which implies that $\mu_{m,m,0}^{(j)} = 1$ for all $j,m$. Hence,
\begin{equation}
\label{eq:33}
S_z = \sum_{m = -N/2}^{N/2} m K_{m,m,0}.
\end{equation}
In addition, one finds  
\begin{equation}
\label{eq:34}
  S_x^2 + S_y^2 = \frac{N}{2} \mathit{id}_{\mathcal{H}} + \sum_{i \neq
    j} \left| 0 1 \right\rangle\!\left\langle 1 0\right| ^{(i,j)} +h.c. =
  \sum_m \frac{N}{2} K_{m,m,0} + K_{m,m,1}.
\end{equation}
Combining equations (\ref{eq:33},~\ref{eq:34}) gives 
\begin{equation}
\label{eq:35}
S^2 = \sum_m  (N/2 + m^2) K_{m,m,0} + K_{m,m,1}.
\end{equation}
As $\langle S^2 \rangle_{\left| j,m,\alpha \right\rangle }  = j(j+1)$, one finds that $m^2 + N/2 + \mu^{(j)}_{m,m,1} = j(j+1)$. Hence, one can calculate the values of $\mu^{(j)}_{m,m,1}$ for all $j,m$. 

Multiplying $X_{(N/2+m-d,d,d,N/2-m-d)}$ by $\pi_g X_{(N/2+m-1,1,1,N/2-m+1)} \pi_g^{\dagger}$, for some $g \in S_N$ and 
$d>1$, results in one out of four possibilities, depending on the choice of $g$
\begin{equation}
\label{eq:400}
X_{(N/2+m-d,d,d,N/2-m-d)}\pi_g X_{(N/2+m-1,1,1,N/2-m+1)} \pi_g^{\dagger} 
\rightarrow
\begin{cases}
   X_{(N/2+m-d-1,d+1,d+1,N/2-m-d-1)}\\
   X_{(N/2+m-d,d,d,N/2-m-d)}\\
   X_{(N/2+m-d+1,d-1,d-1,N/2-m-d+1)}\\
   0.
\end{cases}
\end{equation}
To carry over the results for the respective symmetrised versions, one 
has to carefully count the number of nontrivial permutations of 
$K_{m,m,d}$ and $K_{m,m,1}$ and the number of instances for the four 
cases in equation \eqref{eq:400}. For example, the third possibility in 
equation \eqref{eq:400} occurs $N!/\left[ (N/2+m-d)! d!d! (N/2-m-d)! \right] 
(N/2 + m - d) (N/2 - m -d)$ times, which is $(d+1)^2$ times more often 
than the total number of nontrivial permutations for $K_{m,m,d+1}$. Similarly, one 
has to reason for the other two nontrivial cases. Then, one has
\begin{equation}
\label{eq:36}
K_{m,m,d}K_{m,m,1} = \left[ \left(\frac{N}{2}-d+1\right)^2-m^2 \right]
  K_{m,m,d-1}  +d(N-2d)K_{m,m,d}+(d+1)^2 K_{m,m,d+1}.
\end{equation}
Hence by taking successive powers of $S^2$,
\begin{equation}
\label{eq:37}
S^{2n} = \sum_{m=-N/2}^{N/2}\sum_{d = 0}^n \kappa^{(n,0)}_{m,m,d} K_{m,m,d}
\end{equation}
(the index $0$ will be used later) one obtains for any $n$
\begin{equation}
\label{eq:38}
\left[ j \left( j+1 \right) \right]^n =  \sum_{m=-N/2}^{N/2}\sum_{d = 0}^n \kappa^{(n,0)}_{m,m,d} \mu^{(j)}_{m,m,d}.
\end{equation}
Starting with $n=1$, one calculates $\mu^{(j)}_{m,m,n}$ recursively by inserting the eigenvalues $\mu^{(j)}_{m,m,d}$ with $d < n$.

In the case of $m \neq m^{\prime}$, observe that $S_{-} = S_x - i S_y = \sum_{m= -N/2}^{N/2-1} K_{m,m+1,1/2}$.   Similar arguments as for equation \eqref{eq:36} lead to
\begin{equation}
\label{eq:39}
  K_{m,m^{\prime},d}K_{m^{\prime},m^{\prime}+1,1/2} = \left[ \frac{1}{2}(N+m+m^{\prime}) - d + 1 \right] K_{m,m^{\prime} + 1, d-1/2}  + \left[ \frac{1}{2}(m+m^{\prime})+ d +1 \right] K_{m,m^{\prime}+1,d+1/2}.
\end{equation}
Hence, one iteratively finds the expression of $S^{2n}S_{-}^i$ in terms of the Hartmann basis, equation \eqref{eq:27}),
\begin{equation}
\label{eq:40}
S^{2n}S_{-}^i = \sum_{m=-N/2}^{N/2-i}\sum_{d = i/2}^{n+i/2} \kappa^{(n,i)}_{m,m,d} K_{m,m+i,d}.
\end{equation}
As  $\left\langle j,m,\alpha \right| S^{2n}S^i_{-} \left| j,m+1,\alpha \right\rangle  = [j(j+1)]^n \prod_{k=1}^i C^{-}_{j,m+k}$, with $C^{-}_{j,m} = \sqrt{j(j+1)-m(m-1)}$, one has linear equations for any $(n,i)$ which can be solved recursively using the solution for $(n-1,i)$:
\begin{equation}
\label{eq:41}
\sum_{d = i/2}^{n+i/2} \kappa^{(n,i)}_{m,m,d} \mu^{(j)}_{m,m+i,d} = [j(j+1)]^n \prod_{k=1}^i C^{-}_{j,m+k}.
\end{equation}

In this way, one finds all $\mu^{(j)}_{m,m^{\prime},d}$ and therefore an efficient way to express $\rho$ in the basis of $P^{(j)}_{m,m^{\prime}}$. The actual implementation of the effective density operator description is done with four-dimensional tensors. Their processing is facilitated by using the Tensor Toolbox~\cite{TTB_Software}. The QFI is then calculated via the numerical diagonalising of $\rho$.

\subsection{Remarks on other scenarios}

We shortly note how the method changes if one considers different scenarios.

\textit{Depolarisation noise.---} In section~\ref{sec:local-depolarization}, local depolarisation noise is considered. The action of depolarising noise on diagonal elements is: $\mathcal{E}(\left| 0 \right\rangle\!\left\langle 0\right| ) = (1+e^{-\gamma t})/2 \left| 0 \right\rangle\!\left\langle 0\right|  + (1-e^{-\gamma t})/2 \left| 1 \right\rangle\!\left\langle 1\right| $ and  $\mathcal{E}(\left|1 \right\rangle\!\left\langle 1\right| ) = (1+e^{-\gamma t})/2 \left| 1 \right\rangle\!\left\langle 1\right|  + (1-e^{-\gamma t})/2 \left| 0 \right\rangle\!\left\langle 0\right| $. Using this fact it is straightforward to determine the action of depolarising noise on the Hartmann operators as
\begin{equation}
\label{eq:11}
\mathcal{E}^{\otimes N} K_{m,m^{\prime},d} = \sum_{\nu = -n/2 - i/2}^{n/2 - i/2} e^{-2 d \gamma t} \beta_{\nu} K_{\nu,\nu+i,d},
\end{equation}
where $n = N/2-d$, $i = m^{\prime} - m$ and (with $\epsilon = \nu - m$)
\begin{equation}
\label{eq:12}
  \beta_{\nu} = \left( \frac{1+e^{-\gamma t}}{2} \right)^n\times
  \begin{cases}
\frac{\binom{i_1+\epsilon}{\epsilon}}{\binom{i_4}{\epsilon}} \sum_{x \in X^{+}} \binom{i_1}{x} \binom{i_4}{x+\epsilon}  (\frac{1-e^{-\gamma t}}{1+e^{-\gamma t}})^{\epsilon + 2x}   & \epsilon \geq 0\\
\frac{\binom{i_4-\epsilon}{-\epsilon}}{\binom{i_1}{-\epsilon}} \sum_{x \in X^{-}} \binom{i_1}{x-\epsilon} \binom{i_4}{x}  (\frac{1-e^{-\gamma t}}{1+e^{-\gamma t}})^{-\epsilon + 2x}      & \epsilon < 0.
  \end{cases}
\end{equation}
The values of $i_k$ are determined by $(m,m^{\prime},d)$ (see equation \eqref{eq:27}) and the sets for the summation index read
\begin{equation}
\label{eq:13}
  X^{+}  = \left\{ 0, \dots, \min(i_1, i_4 - \epsilon) \right\} \quad \text{and} \quad
  X^{-}  = \left\{ 0, \dots, \min(i_4, i_1 + \epsilon) \right\}.
\end{equation}
The effort in calculating the action of depolarisation noise is thus higher than that of local phase noise. For this reason the maximal system size is restricted to $N = 30$.

\textit{Correlated noise.---} In the presence of correlated noise (see section~\ref{sec:correlated-noise}), the numerical treatment is straightforward. The action onto an off-diagonal element simply reads $\mathcal{E}| N/2,m \rangle\!\langle N/2,m^{\prime}| = p^{(m-m^{\prime})^2}| N/2,m \rangle\!\langle N/2,m^{\prime}|$ with $p = \exp( - \gamma t - e^{-\gamma t} + 1)$. This means that the noisy state stays within the subspace spanned by $\left\{ \left| N/2,m \right\rangle  \right\}_m$. Therefore, a more complicated treatment using the Hartmann basis is not necessary.

\section{Asymptotic behaviour}
\label{sec:asymptotic-behavior}
\label{sec:deta-some-analyt}
\label{sec:langle-s_x-rangle}

As discussed in section \ref{sec:spin-squeezing-not}, the numerical results for the standard scenario suggest that the variance as well as a specific behaviour of the coefficients, $c_m$, of the state, in the eigenbasis of the Hamiltonian, play a key role in the metrological performance of a quantum state in the standard scenario. An interesting question is what happens to the average improvement of $\mathfrak{F}$ for certain state families as $N$ becomes large. For instance, although the relative difference between $\mathfrak{F}_{\mathrm{opt}}$ and $\mathfrak{F}$ of Dicke states increases in the investigated range of $N$, this gap could stabilise to a finite value, or even close in the limit $N\to\infty$.  Unfortunately, in this limit, numerical results do not provide definitive answers as calculation of the QFI requires the spectral decomposition of the time-evolved quantum state. Even for simple state families, this is asymptotically not feasible.

In this appendix, we consider a particular measurement that clearly leads to a lower bound on the QFI. The measurement is assumed to be in the eigenbasis of $S_\varphi = \exp(-i \varphi S_z) S_x \exp(i \varphi S_z)$, i.e.~a measurement in the $x-y$ plane. The quantity
\begin{equation}
\label{eq:7}
G = \max_{\varphi} \frac{\left|  \frac{d} { d \omega} \langle S_\varphi \rangle_{\rho(t)}\right|^2} {\mathcal{V}_{\rho(t)}(S_\varphi)},
\end{equation}
which has also been considered in~\cite{huelga_improvement_1997, ulam-orgikh_spin_2001}, is equivalent to the QFI for any initial state if $N = 1$, and for the product state for any $N$. It was recently shown~\cite{zhong_optimal_2013} that $G$ is in general a lower bound on the QFI, i.e.~$G\leq \mathcal{F}$.

Optimising over $\varphi$ and inserting the time-evolved state (see equation \eqref{eq:2}) for a symmetric state, $\ket{\psi}$, we obtain
\begin{equation}
\label{eq:8}
G =  \frac{N t^2 s^2}{e^{2 \gamma t} - 1 + 4 \mathcal{V}_{\psi}(S_y)/N },
\end{equation}
with $s = 2\langle S_{x} \rangle_{\psi }/N$. 

If $s \rightarrow 1$  and $ \mathcal{V}_{\psi}(S_y)/N\rightarrow 0$ for $N \rightarrow \infty$, $\max_t T G/t \leq \mathfrak{F}$ reaches the optimal value $NT/(2\gamma)$ for $t \rightarrow 0$.
We now show that states $| \psi \rangle $ whose distribution of coefficients, $\{c_m\}$, is sufficiently smooth and for which $\mathcal{V}(h) = O(N^2)$, cannot reach $s = 1$ asymptotically. Therefore, such a state can not achieve the bound $NT/(2 \gamma)$, and is suboptimal with respect to the quantity $G$. Note that if a similar statement was to hold true for $\mathfrak{F}$, one could conclude that states satisfying $\mathcal{V}(h)=\mathcal{O}(N^2)$ are asymptotically less useful than the optimal states.

Let us now come to the lemma. 
\begin{Lem}
We consider states $| \psi \rangle \in \mathcal{S}_N$ given in equation (\ref{eq:5}), where the coefficients $c_m = c_{N-m} \in [0,1]$ only depend on $N$. We assume that the coefficients are sufficiently ``smooth'' in the sense that, $c_{m+1} = c_m - \epsilon_m(N)$ with $\lim_{N \rightarrow \infty} \epsilon_m(N) = 0$ for all $m \in \left\{ -N/2,\dots,N/2-1 \right\}$. In addition, be $\left| \psi \right\rangle $ such that $\mathcal{V}_{\psi}(h) = \sum_{i = 0}^N m^2 c_m^2 = O(N^2)$. Then it holds that $2\lim_{N \rightarrow \infty} \langle S_x \rangle_{\psi}/N < 1 $.
\end{Lem}

\begin{proof}
We first compute $\langle S_x \rangle_{\psi}$.  From the theory of angular momentum $S_x=1/2(S_++S_-)$ where the operators $S_{\pm}$ are the standard raising and lower operators whose action on the basis states $\{\ket{j,m}\}$ is given by
\begin{equation}
S_{\pm}\ket{j,m}=\sqrt{(j\mp m)(j\pm m+1)}\ket{j,m\pm 1}
\label{ladders}
\end{equation}
respectively. Computing $\langle S_\pm\rangle_{\psi}$ yields
\begin{equation}
\label{eq:47}
\begin{split}
\langle S_+\rangle_{\psi}&=\sum_{m=-N/2}^{N/2-1}\sqrt{(N/2-m)(N/2+m+1)}c_{m}c_{m+1}\\
\langle S_-\rangle_{\psi}&=\sum_{m=-N/2+1}^{N/2}\sqrt{(N/2+m)(N/2-m+1)}c_{m}c_{m-1}.
\end{split}
\end{equation}
However, as $c_m=c_{N-m} \in \mathbbm{R}$ it is easy to see that $\langle S_+\rangle=\langle S_-\rangle$. Hence
\begin{equation}
  \begin{split}
    \label{eq:46}
    2 \langle S_x \rangle_{\psi}/N&=2\sum_{m = -N/2}^{N/2-1}
    \frac{\sqrt{(\frac{N}{2}-m)(\frac{N}{2}+m+1)}}{N} c_m c_{m+1}\\
    &=\sum_{m = -N/2}^{N/2-1}\sqrt{\left(1-\frac{2m}{N}\right)\left(1+\frac{2m}{N}+\frac{2}{N}\right)} c_m \left[ c_m - \epsilon_m(N)\right],
  \end{split}
\end{equation}
where we have used the smoothness of the coefficients $c_m$ in the second line of equation \eqref{eq:46}.

In order to convert the sum into an integral, let $x = 2m/N$. Taking the limit $N \rightarrow \infty$ one obtains
\begin{equation}
\label{eq:48}
\lim_{N \rightarrow \infty} 2 \langle S_x \rangle_{\psi}/N= \int_{-1}^1 \sqrt{1 - x^2} c(x)^2 dx,
\end{equation}
where $c(x) = \lim_{N \rightarrow \infty} c_{x N/2} \in [0,1]$ with $\int_{-1}^1 c(x)^2 dx = 1$. Note that $\sqrt{1-x^2} \leq 1$ and equals unity if and only if $x = 0$.  Due to the normalisation of $c(x)$, we find that $ \lim_{N \rightarrow \infty} 2 \langle S_x \rangle_{\psi}/N= 1$ only if $c(x)^2 = \delta(x)$, the Dirac delta function. However,  as $\mathcal{V}(h) = O(N^2)$ by assumption, $c(x)^2$ cannot be the Delta function. Hence, $\lim_{N \rightarrow \infty}2 \langle S_x \rangle_{\psi}/N < 1$.  This completes the proof.
\end{proof}

\section{Maximal QFI for GHZ under depolarising noise}
\label{sec:maximal-qfi-ghz}

We consider here the depolarising channel (see equation \eqref{localnoise}).
In~\cite{dur_improved_2013} (Appendix I), we have shown that the QFI of for the GHZ state as initial state is given by
\begin{equation}
\label{eq:22}
\mathcal{F} = t^2 N^2 \frac{2^N e^{-2 \gamma N t}}{\left( 1+e^{-\gamma t} \right)^N + \left( 1-e^{-\gamma t}\right)^N}.
\end{equation}
For the time optimisation, we assume that, for large $N$, one can restrict oneself to $\gamma t \ll 1$, which is \textit{a posteriori} justified. Then, one approximates $\left(1+e^{-\gamma t}\right)/2$ by $e^{-\gamma t/2}$ and neglect $\left[ \left(1-e^{-\gamma t}\right)/2 \right]^N$. Hence, the QFI equals $\mathcal{F} \approx t^2 N^2 \exp(-3/2 N \gamma t)$. The small difference in the power then gives rise to a modified optimal time $t_{\mathrm{opt}} \approx 2/(3 N \gamma)$ which leads to $\mathfrak{F}_{\mathrm{GHZ}} = 2NT/(3 \gamma e)$.

\section*{References}
\label{sec:references}

\bibliographystyle{iopart-num}

\bibliography{References}

\end{document}